\documentclass[draft]{llncs}
\usepackage{ifpdf}
\usepackage{color}               
\usepackage{amsmath,amssymb}
\usepackage[latin1]{inputenc}
\usepackage[OT1]{fontenc}
\usepackage{comment}

\usepackage[all]{xy}

\newtheorem{theo}{Theorem}

\newenvironment{myproof}[1]{
\vspace{-5ex}

\hspace*{1em}\begin{trivlist}\item[\hskip\labelsep{\textsc{#1}}]}%
{\end{trivlist}}
\renewenvironment{proof}{
\begin{myproof}{Proof.}}{\hfill $\Box$\end{myproof}}

\newenvironment{proofclaim}{%
\vspace{-5ex}
\addtolength{\linewidth}{-1ex}
\begin{center}
\hspace*{1ex}\begin{trivlist}
\item\begin{myproof}{Proof of the claim.}}{
$\Diamond$\end{myproof}
\end{trivlist}
\end{center}
\addtolength{\linewidth}{1ex}
\usepackage[latin1]{inputenc}
\usepackage[OT1]{fontenc}
\usepackage{comment}

\vspace{-2ex}
}

\newcommand{\call}[1]{\ensuremath{{\mathcal{#1}}}}

\newcommand{\mathconstante}[1]{\ensuremath{{\mathrm{#1}}}}

\newcommand{\Variables}{\call{X}}
\newcommand{\Constants}{\mathconstante{C}}
\newcommand{\fonction}[2]{\ensuremath{{\mathconstante{#1}(#2)}}}

\newcommand{\set}[1]{\ensuremath{{\left\lbrace #1 \right\rbrace}}}

\newcommand{\termset}[1]{{\fonction{T}{#1}}}

\newcommand{\rhnorm}[2]{\ensuremath{{(#1)\!\!\downarrow_{#2}}}}

\newcommand{\norm}[1]{\rhnorm{#1}{}}

\newcommand{\Var}[1]{\fonction{Var}{#1}}
\newcommand{\Sub}[1]{\fonction{Sub}{#1}}

\newcommand{\Cons}[1]{\fonction{trace}{#1}}

\newcommand{\Supp}[1]{\fonction{Supp}{#1}}

\newcommand{\tq}{\ensuremath{\,|\,}}
\newcommand{\condset}[2]{\set{#1\tq{}#2}}

\newcommand{\rhclos}[2]{{\ensuremath{\overline{#1}^{#2}}}}
\newcommand{\clos}[1]{\rhclos{#1}{}}

\newenvironment{decisionproblem}[1]{
\vspace*{-1ex}
\begin{tabbing}
  \underline{\textbf{#1}}\\
  \hspace*{2em}\= \textbf{Input:}~ \=}{
\end{tabbing}
\vspace*{-1em}
}

\newcommand{\entreeu}[1]{ \hbox{\vbox{\parbox[t]{0.8\linewidth}{#1}}}\\}
  
\newcommand{\sortie}[1]{
\> \textbf{Output:} \> \hbox{\vbox{\parbox[t]{0.8\linewidth}{#1}}}\\}

\newcommand{\unif}{\ensuremath{\stackrel{?}{=}}}
\newcommand{\intrus}[3]{\ensuremath{\left\langle #1,#2,#3\right\rangle}}

\newcommand{\II}{\ensuremath{\call{I}'_{\mbox{\tiny\rm }}}}
\newcommand{\nnintrus}[1]{\intrus{\call{G}}{S}{\emptyset}}

\newcommand{\nintrus}[1]{\intrus{\call{G}}{S}{\call{H}}}

\newcommand{\gsig}[1]{\termset{#1}}
\newcommand{\vsig}[1]{\termset{#1,\Variables}}

\newcommand{\TGX}{\vsig{\call{G}}}

\newcounter{lemc}
\newcommand{\rew}{\rightarrow}

\newcommand{\ded}{\twoheadrightarrow}

\newcommand{\M}[1]{\ensuremath{\mathop{\mbox{\rm M}}(#1)}}

\newcommand{\nbv}[1]{\ensuremath{|\Var{#1}|}}

\newcommand{\SSub}[1]{\ensuremath{\mathop{\mbox{\rm SSub}}(#1)}}
\title{On the Decidability of (ground) Reachability Problems for Cryptographic
Protocols (extended version)}
\author{Yannick Chevalier\inst{1} \and Mounira Kourjieh\inst{2}}
\institute{INRIA Nancy Grand Est, Loria, France, email:\texttt{Yannick.Chevalier@loria.fr}
\and IRIT, Universit\'e de Toulouse, France, email: \texttt{Mounira.Kourjieh@irit.fr}}
\date{}
\usepackage{proof}
\begin{document}
\sloppy
\maketitle{}
\begin{abstract}
  Analysis of cryptographic protocols in a symbolic model is relative to a
  deduction system that models the possible actions of an attacker regarding an
  execution of this protocol. We present in this paper a transformation
  algorithm for such deduction systems provided the equational theory has the
  finite variant property. the termination of this transformation entails the
  decidability of the ground reachability problems. We prove that it is
  necessary to add one other condition to obtain the decidability of non-ground
  problems, and provide one new such criterion.
\end{abstract}

\section{Introduction}\label{sec:intro}
Cryptographic protocols are programs designed to ensure secure electronic communications between participants using  an  insecure
networks. 
Unfortunately, the existence of cryptographic primitives such as encryption and digital signature is not sufficient to ensure security and several attacks were found on established protocols \cite{CJ97,spore}. 
The most relevant example is the bug of the Needham-Schroeder protocol found by Lowe \cite{Lowe95} using a model-checking tool. It took 17 years since the protocol was published to find the attack, a so-called \textit{man-in-the-middle} attack.
This situation leads to the development of tools and decision
procedures for the formal verification of security protocols.  There are several
approaches to modelling 
cryptographic protocols and analysing their security properties: reachability analysis (e.g.: NRL \cite{meadows96}), model checking (FDR~\cite{Lowe97,Lowe99a}, Mur$_\varphi$
\cite{MMS97}), modal logic and deduction \cite{BurrowsAN90}, process calculi like \textit{the spi-calculus} \cite{AbadiG99}, so-called cryptographic proofs (\cite{AbadiR00}) and others.
%
 Here, we use yet another
technique, based on the resolution of reachability problems.

Early works on verification of cryptographic protocols studied the standard
Dolev-Yao intruder model~\cite{RusinowitchT-TCS03} and the perfect
cryptography~\cite{dolev83ieee} which states that it is impossible to obtain any information about an encrypted message without knowing the exact key necessary to decrypt this message.  Unfortunately, this perfect cryptography assumption has proven too idealistic: there are protocols which can be proven secure under perfect cryptography assumption, but which are in reality insecure since an attacker can use properties of the cryptographic primitives
in combinaison with the protocol rules in order to attack protocol. These properties (so-called algebraic properties) are typically expressed as equational theories.
An overview on algebraic properties of well-known cryptographic primitives can be found in \cite{CortierDL}.
%
 In this paper, we study the class of
equational theories represented by a finite convergent rewrite system and having the
finite variant property modulo the empty theory~\cite{Comon-LundhD05}. 

Another point of interest is that an intruder is modelled by a deduction system
representing the possible inferences it can make on the messages it knows. A
\emph{ground reachability problem} for a given deduction system consists in
giving a proof using the permitted deductions of a fact represented by a ground
term $t$ from a set of known facts represented by a finite set of terms $E$.
\emph{General reachability problems} are generalisation of the problem in which
the goal $t$ has non variables, and the goal is to find a ground substitution
$\sigma$ of these variables such that the instance $t\sigma$ is provable from a
finite set of ground terms $E$. This generalisation consists in providing
intermediate steps to solve.  

\paragraph{Proof strategy.} In~\cite{fossacs04}, H.~Comon-Lundh
proposes a two-steps strategy to solve general reachability problems,
\textit{i.e.} first to solve the ground reachability problems by invoking some
locality argument, and then to reduce general reachability problems to ground
ones. The method described in this paper roughly follows this line.  We employ
the finite variant property to reduce reachability problems modulo an equational
theory to reachability problems modulo the empty theory. We then partially
compute a transitive closure of the possible deductions. We prove that the
termination of this computation implies the decidability of the ground
reachability problems.  We
conjecture that the overall construction amounts to proving that the deduction
system is $F$-local~\cite{BernatC06}. We then give a new criterion that permits us to reduce general
reachability problems to ground reachability problems. This criterion is based
on counting the number of variables in a reachability problem before and after a
deduction is guessed, and is a generalisation of the one employed for the
specific case of the Dolev-Yao intruder model. The intuition behind this
criterion is that a deduction rule has to provide more relations between
existing fact than it introduces new unknown. We give an example showing that
such an additional criterion is needed, in the sense that there exists deduction
systems on which the saturation algorithm terminates, but for which the general
reachability problems are undecidable.
Another contribution of this paper is a decidability result to   the ground reachability problems for the theory of blind signature \cite{KremerR05} using the initial definition of subterm introduced in \cite{AbadiC06,Baudet05},
a similar result was given in  \cite{AbadiC05} using an extended definition of subterm. 
In addition we give a decidability result to the general reachability problems for a class of subterm convergent equational theories, while a more general  result was given in  \cite{Baudet05}, the proof given in this paper  for our special case is much shorter.

\paragraph{Related works.} Several decidability results have been obtained for cryptographic protocols in a
similar setting~\cite{AmadioL00,MillenS01,AmadioLV03}.  These results have
been extended to handle algebraic properties of cryptographic primitives
\cite{ChevalierKRT-FSTTCS03,ChevalierKRT05,ChevalierK07,AbadiC05}.  In~\cite{AbadiC06}, a decidability result was given to the ground
reachability problems in the case of subterm convergent equational theories. 
This  result  was extended
in~\cite{AbadiC05} and a
decidability result to the ground reachability problems in the case of locally
stable AC-convergent equational theories was given. Moreover, again
in~\cite{AbadiC05}, a decidability result was given to the ground reachability
problems for the theory of blind signature \cite{KremerR05} while  this theory was not included in  \cite{AbadiC06}.
To obtain a decidability result for the theory of blind signature, Abadi and Cortier \cite{AbadiC05} use a new extended definition of subterm. The result obtained in \cite{AbadiC06} was extended in \cite{Baudet05} in  different way than in \cite{AbadiC05} and a decidability result was obtained to the general reachability problem for the class of subterm convergent equational theory.
The first result of our paper is a decidability result to the ground reachability problems for a class of equational theories which includes the class studied in \cite{AbadiC06}.  We note  that the class studied in \cite{AbadiC05} is  incomparable with ours and we note also  that  the proof used in \cite{AbadiC05} to decide the ground reachability problems for the theory of blind signature is different from the ours.
 Another result of this paper is a decidability result to the general reachability problem for a class  of equational theories under some conditions on the deduction systems and the class studied in \cite{Baudet05} is incomparable with ours.
In \cite{BernatC06}, a
decidability result was given to the general reachability problems under some
syntactic conditions on the intruder deduction rules, this result is incomparable with ours.

\section{Preliminaries} \label{sec:setting} 

We now introduce some notations and basic definitions for terms, equational
theories and term rewriting systems (the reader may refer to
\cite{DershowitzJ90} for more details), and then proceed with the definition
of the so-called intruder constraints.

\subsection{Terms}
We assume given a signature $\call{G}$, an infinite set of variables $\call{X}$
and an infinite set of free constants \Constants{}. The set of terms built with
$\call{G}$ and \Variables{} is denoted \vsig{\call{G}} and its subset of ground
terms (terms without variables) \gsig{\call{G}}.  We denote \Var{t} the set of variables occurring in a
term $t\in\vsig{\call{G}}$, $\nbv{t}$ the number of elements  in the set $\nbv{t}$ that is the number of distinct variables occurring
in $t$, $\Sub{t}$ the set of subterms of $t$ and $\SSub{t}$ the set of strict
subterms of $t$. These notations are extended as expected to sets of terms.   We denote $t[s]$ a term $t$ that admits $s$ as
subterm and $t[s\leftarrow s']$ the term  $t$ in which $s$ is replaced by $s'$.


A substitution $\sigma$ is an involutive mapping from \Variables{} to
\vsig{\call{G}} such that $\Supp{\sigma}=\condset{x}{\sigma(x)\not=x}$, the
\emph{support} of $\sigma$, is a finite set.  The application of a substitution
$\sigma$ to a term $t$ (resp.  a set of terms $E$) is denoted $t\sigma$
(resp. $E\sigma$).  A substitution $\sigma $ is \emph{ground} w.r.t.  $\call{G}$
if the image of $\Supp{\sigma }$ is included in $\gsig{\call{G}}$.

We recall in the following the definition of  \textit{reduction order}:
\begin{definition}\label{def:reduction-order}
Let $\call{G}$ be a signature and \call{X} be an infinite set of variables. A strict order $\succ$ on $\vsig{\call{G}}$ is called a rewrite order iff it is
\begin{enumerate}
\item \textbf{compatible} with $\call{G}$-function symbols: for all $s,s'\in\vsig{\call{G}}$ and all $f\in\call{G}$ with arity $n\geq 0$, $t_1\succ t_2$ implies 
 $$ f(t_1,\ldots,t_{i-1},s,t_{i+1},\ldots,t_n)\succ f(t_1,\ldots,t_{i-1},s',t_{i+1},\ldots,t_n)$$ 
for all $i$, $1\leq i\leq n$, and all $t_1,\ldots,t_{i-1},t_{i+1},\ldots,t_n\in\vsig{\call{G}}$.
\item
\textbf{closed under substitutions:} for all $s,s'\in\vsig{\call{G}}$ and all substitutions $\sigma$, $s\succ s'$
implies $\sigma(s)\succ \sigma(s')$.
\end{enumerate}
A \textbf{reduction order} is a well-founded rewrite order.
\end{definition}

 We consider
a reduction order $\succ$ over $\vsig{\call{G}}$ total over ground terms.  We
denote $\succeq$ the relation between terms such that $t_1\succeq t_2$ iff
$t_1\succ t_2$ or $t_1=t_2$ for $t_1,t_2\in\vsig{\call{G}}$.

A rewriting system \call{R} is a finite set of couples $(l,r)\in\vsig{\call{G}}^2$, where
each couple is called a \emph{rewriting rule} and is denoted $l\rew r$. The
rewriting relation $\rew_\call{R}$ between terms is defined by $t\rew_\call{R}
t'$ if there exists $l\to r\in \call{R}$ and a substitution $\sigma$ such that
$l\sigma=s$ and $r\sigma=s'$, $t=t[s]$ and $t'=t[s\leftarrow s']$. A rewriting
system is \emph{terminating} if for all terms $t$ there is no infinite sequence
of rewriting starting from $t$. It is \emph{convergent} if it has moreover the
\emph{confluence} property: every sequence of rewriting ends in the same term
denoted $\norm{t}_{\call{R}}$, or simply \norm{t} if \call{R} is clear from the
context. We say that a term $t$ is in \emph{normal form} if
$t=\norm{t}_\call{R}$. A substitution $\sigma$ is in normal form if for all
$x\in\Supp{\sigma}$, the term $\sigma(x)$ is in normal form. Given a
substitution $\sigma$, we denote $\norm{\sigma}_\call{R}$ the substitution such
that, for all $x\in\Supp{\sigma}$ we have
$\norm{x\sigma}_\call{R}=x\norm{\sigma}_\call{R}$.

An equational theory \call{H} is a congruence relation on terms in
\vsig{\call{G}}. We denote $t=_\call{H}t'$ the fact that the term $t$ and $t'$
are identified by \call{H}.  We say that \call{H} is \emph{generated} by a
convergent rewriting system \call{R} if $t=_\call{H}t'$ iff
$\norm{t}_\call{R}=\norm{t'}_\call{R}$.

\subsection{Unification systems}


\begin{definition}{\label{def:unification}(Unification systems)}
  Let \call{H} be an equational theory.  A $\call{H}$-\emph{unification system}
  \call{S} is a finite set of pairs of terms in \TGX{} denoted by $\set{u_i
    \unif{}_{\call{H}} v_i}_{i\in\set{1,\ldots,n}}$.  It is satisfied by a
  substitution $\sigma$, and we note $\sigma\models{}_{\call H}\call{S}$, if for
  all $i\in\set{1,\ldots,n}$ we have $u_i\sigma =_\call{H} v_i\sigma$.  In this
  case we call $\sigma$ a \emph{solution} or a \emph{unifier} of \call{S}.
\end{definition}

When \call{H} is generated by a convergent rewriting system \call{R},
considering a bottom-up normalisation shows that if $\sigma$ is a solution of a
\call{H}-unification system, then \norm{\sigma} is also a solution of the same
unification system. A top-down normalisation on solutions also demonstrates that
we can assume that terms in a unification system are in normal form. Accordingly
we will consider in this paper only solutions in normal form of unification
systems in normal form. A unifier $\sigma$ is more general than a unifier $\tau$
if there exists a substitution $\theta$ such that $\sigma\theta=\tau$.  A
\emph{complete set of unifiers} of a \call{H}-unification system \call{S} is a
set $\Sigma$ of unifiers of \call{S} such that, for any unifier $\tau$ of
\call{S}, there exists $\sigma\in\Sigma$ which is more general than $\tau$. The
unifier $\tau$ is a \emph{most general unifier} of \call{S} if the substitution
$\theta$ in the preceding equation is a variable renaming. We denote
$mgu(\call{S})$ the set of most general unifiers modulo \call{H} of a
unification system \call{S}.  In the context of unification modulo an equational
theory, standard (or syntactic) unification will also be called unification in
the empty theory. In this case, it is well-known that there exists a unique most
general unifier of a set of equations. This unifier is denoted $mgu(\call{S})$,
or $mgu(s,t)$ in the case $\call{S}=\set{s\unif_\emptyset t}$.

\paragraph{Finite Variant Property.}
We will abusively write that an equational theory \call{H} has the \emph{finite
  variant property} if the couple $(\call{H},\emptyset)$ has the finite variant
property in the notation of~\cite{Comon-LundhD05}. Let us now formally state the
definition of this property in this case, simplified using the Lemma~3 and the
Theorem~1 of~\cite{Comon-LundhD05}.

\begin{definition}{(Finite Variant Property)\label{def:FVP}}
  A theory \call{H} has the \emph{finite variant property} if, for any term $t$,
  one can compute a finite set of substitutions $\theta_1,\ldots,\theta_n$ (the
  variant substitutions) such that, for any substitution $\sigma$ in normal form
  there exists $i\in\set{1,\ldots,n}$ and a substitution $\sigma'$ in normal
  form such that $\sigma=\theta_i\sigma'$ and
  $\norm{t\sigma}=\norm{t\theta_i}\sigma'$. The terms $\norm{t\theta_i}$ are
  called the \emph{variants} of $t$.
\end{definition}

Examples of equational theories  having the finite variant property
are those defined by a convergent rewriting system and  such that either basic
narrowing~\cite{Hullot80} terminates or  the rewriting system is
optimally reducing~\cite{NarendranPS97}. 

The finite variant property ensures that it is possible to compute a complete
set of most general unifiers between two terms $t$ and $t'$. Indeed, it suffices
to compute for these two terms the respective sets of variant substitutions
$\set{\theta_i}_{i\in\set{1,\ldots,m}},\set{\theta'_j}_{j\in\set{1,\ldots,n}}$,
and to (try to) unify in the empty theory every pair of terms
$\norm{t\theta_i}\unif\norm{t'\theta'_j}$.

\begin{center}
  \fbox{\vspace*{8em} ~~\parbox{0.92\linewidth}{In the rest of this paper we
      will consider equational theories \call{H} having the finite variant
      property and generated by a convergent rewriting system \call{R}.}~~}
\end{center}

\subsection{Deduction systems}
The notions that we give here have been defined in~\cite{ChevalierR05}.  These
definitions have since been generalised to consider a wider class of intruder
deduction and constraint systems~\cite{ChevalierLR07}. Although this general class
encompasses all deduction and constraint systems given in this paper, we have
preferred to give the simpler definitions from~\cite{ChevalierR05} which are
sufficient for stating our problem.  We will refer, without further
justifications, to the model of~\cite{ChevalierLR07} as \emph{extended} deduction
systems. The constraint systems considered and defined here correspond to
symbolic derivations~\cite{ChevalierLR07} in which a most general unifier of the
unification system has been applied on the output messages (for
Def.~\ref{def:constraints}) and on input variables (for the extended constraint
systems).

In the context of a security protocol (see \textit{e.g.}~\cite{meadows96} for a
brief overview), we model messages as ground terms and intruder deduction rules
as rewriting rules on sets of messages representing the knowledge of an
intruder.  The intruder derives new messages from a given (finite) set of
messages by applying deduction rules.  Since we assume some equational axioms
$\call{H}$ are satisfied by the function symbols in the signature, all these
derivations have to be considered \emph{modulo} the equational theory $\call{H}$
generated by $\call{R}$.

\begin{definition}{\label{def:intruder}}
  A \emph{deduction system} $\call I$ is given by a triple
  \intrus{\call{G}}{\call{L}}{\call{H}} where $\call{G}$ is a signature,
  $\call{L}$ is a set of deduction rules $l\ded r$, where $l$ a set of terms in $\vsig{\call{G}}$ and $r$ a term in $\vsig{\call{G}}$,  and $\call H$ is an equational theory.
\end{definition}

Each rule $l\ded{}r$ in $\call{L}$ defines a deduction relation
$\ded_{l\ded{}r}$ between finite sets of terms.  Given two finite sets of terms
$E$ and $F$ we have $E\ded_{l\ded{}r}F$ if and only if there exits a
substitution $\sigma$, such that $l\sigma=_{\call{H}}l'$,
$r\sigma=_{\call{H}}r'$, $l'\subseteq{}E$ and $F = E \cup \set{r'}$.  We denote
$\ded_{\call I}$ the union of the relations $\ded_{l\ded{}r}$ for all $l\ded{}r$
in $\call{L}$ and by $\ded_{\call I}^*$ the transitive closure of $\ded_{\call
  I}$.  Note that, given sets of terms $E$, $E'$, $F$ and $F'$ such that
$E=_\call{H}E'$ and $F=_\call{H}F'$ by definition we have $E\ded_{\call I}F$ iff
$E'\ded_{\call I}F'$. We simply denote by $\ded$ the relation $\ded_{\call I}$
when there is no ambiguity about ${\call I}$.

We recall that $\succeq$ is the extension of the reduction order $\succ$ defined over $\vsig{\call{G}}$.
 
\begin{definition}
  A deduction rule $l\ded r$ is a \emph{decreasing} rule if there is a term
  $s\in l$ such that $s\succeq r$ and it is \emph{increasing}
  otherwise. 
\end{definition}

From now, if \call{L} is the set of deduction rules, we denote by
$\call{L}_{inc}$ the set of increasing rules and by $\call{L}_{dec}$ the set of
decreasing rules. By definition of \emph{increasing} and \emph{decreasing}
rules, we have $\call{L}=\call{L}_{inc}\cup\call{L}_{dec}$.

A \emph{derivation} $D$ of length $n$, $n\ge 0$, is a sequence of steps of the
form $E_0\ded_{\call I} E_0,t_1\ded_{\call I}\cdots\ded_{\call I} E_n$ with
finite sets of terms $E_0,\ldots{}E_n$, and terms $t_1,\ldots,t_n$, such that
$E_i=E_{i-1}\cup{}\set{t_i}$ for every $i\in \set{1,\ldots,n}$.  The term $t_n$
is called the {\em goal} of the derivation. We let $\Cons{D}$  be the set of
terms constructed during the derivation $D$,
$\Cons{D}=E_0\cup\set{t_1,\ldots,t_n}$. We define $\rhclos{E}{\call I}$ to be equal
to the set of terms that can be deduced from $E$, $\rhclos{E}{\call I}=\set{t ~ s.t. ~ E\ldots^*_\call{I} E' ~ and ~ t\in ~ E'}$. If there is no ambiguity
on the deduction system $\call I$ we write \clos{E} instead of $\rhclos{E}{\call
  I}$.

\subsection{Constraint systems\label{reachprob}} 

We now introduce the constraint systems to be solved for checking
protocols. It is presented in~\cite{ChevalierR05} how these constraint
systems permit to express the reachability of a state in a protocol
execution.

\begin{definition}{\label{def:constraints}(\call{I}-Constraint systems)} 
  Let ${\call I} =\langle \call{G}, \call{L}, \call{H} \rangle $ be a
  deduction system.  An $\call I$-\emph{constraint system} \call{C} is
  denoted $((E_i\rhd{}v_i)_{i\in\set{1,\ldots,n}},\call{S})$ and is
  defined by a sequence of pairs $(E_i, v_i)_{i\in\set{1,\ldots,n}}$
  with $v_i\in\Variables{}$, $E_i \subseteq \TGX$, $E_i\subseteq E_{i+1}$ and 
  $\Var{E_i}\subseteq\set{v_1,\ldots,v_{i-1}}$ for
  $i\in\set{1,\ldots,n}$, and by an $\call
  H$-unification system \call{S}.
  
  An $\call I$-\emph{Constraint system} \call{C} is satisfied by a
  substitution $\sigma$ if for all $i\in\set{1,\ldots,n}$ we have
  $v_i\sigma\in\rhclos{E_i\sigma}{\call{I}}$ and if
  $\sigma\models_\call{H}{}\call{S}$. We denote that a substitution
  $\sigma$ satisfies a constraint system \call{C} by
  $\sigma\models_{\call I}\call{C}$.
\end{definition}

Constraint systems are denoted by \call{C} and decorations thereof.
Note that if a substitution $\sigma$ is a solution of a constraint
system \call{C}, by definition of deduction rules and unification
systems the substitution \norm{\sigma} is also a solution of \call{C}.
In the context of cryptographic protocols the inclusion
$E_{i-1}\subseteq E_{i}$ means that the knowledge of an intruder does
not decrease as the protocol progresses: after receiving a message a
honest agent will respond to it, this response can then be added to
the knowledge of the intruder who listens to all communications.  The
condition on variables stems from the fact that a message sent at 
step $i$ must be built from previously received messages recorded in
the variables $v_j, j<i $, and from the  initial knowledge (set of ground terms)  of
the honest agents.  Our goal is to solve the following  decision problem.

\begin{decisionproblem}{\call{I}-Reachability Problem}
  \entreeu{ An \call{I}-constraint system \call{C}.}
  \sortie{\textsc{Sat} iff there exists a substitution $\sigma$ such
    that $\sigma\models_{\call I}\call{C}.$}
\end{decisionproblem}

\section{Saturation}\label{sec:sat}

In the rest of this paper, we suppose that
$\call{I}_0=\intrus{\call{G}}{\call{L}_0}{\call{H}}$ is an initial deduction
system. \emph{We assume that $\call{L}_0$ is the union of rules
  $x_1,\ldots,x_n\ded f(x_1,\ldots, x_n)$ for some function symbols
  $f\in\call{G}$}.


Let $\call{H}$ be an equational theory having the finite variant property and
generated by a convergent rewriting system \call{R}.  The saturation of the set
of deduction rules $\call{L}_0$ defined modulo the equational theory \call{H} is
the output of the application of the saturation algorithm given by the following
two steps:

\begin{itemize}
\item \textit{Step ~ 1: } Anticipating the application of rules of $\call{L}_0$
  on ground terms in normal form, we define the set \call{L} of rules ``in
  normal form'':
  $$
  \call{L}=\bigcup_{
    \begin{array}{c}
      x_1,\ldots,x_n\ded f(x_1,\ldots,x_n)\in\call{L}_0\\
      \theta\text{ variant subsitution of }f(x_1,\ldots,x_n)\\
    \end{array}
  } x_1\theta,\ldots,x_n\theta\ded \norm{f(x_1,\ldots,x_n)\theta}
  $$
  This union is over finite sets thanks to the finiteness of $\call{L}_0$ and to
  the finite variant property.
\item \textit{Step ~ 2:} Start with $\mathcal{L}' =\call{L}$, repeat the rule
  given in Figure~\ref{fig:closure} until no new rule can be added.

\begin{figure*}[htbp]
  $$
  \begin{array}{lc}
     \infer
      [\begin{array}{c}
       s\notin\Variables\\  
        \sigma=mgu_\emptyset(r_1,s)\\
      \end{array}]
      {
        \call{L}'\gets\call{L}'\cup\set{(l_1,l_2\ded r_2)\sigma}
      }
      {
        l_1\ded r_1 \in \mathcal{L}'_{inc}~ ; ~~~~l_2,s\ded r_2\in \mathcal{L}'
      }
  \end{array}
  $$
  \vspace*{-1em}
  \caption{\label{fig:closure} closure rule.}
  \vspace*{-2em}  
\end{figure*}
\end{itemize}

We define two new deduction systems, corresponding each to one step of the
saturation algorithm, $\call{I}= \intrus{\call{G}}{\call{L}}{\emptyset}$ and
$\II=\intrus{\call{G}}{\call{L}'}{\emptyset}$.
Since in the first step we consider all possible variants of all possible
deduction rules, we have:

\begin{lemma}{\label{lemma:00:saturation1}}
  Let $E$ and $F$ be two  sets of ground terms in normal form 
 we have: $E\ded_{\call{I}_0} F$ iff 
  $E\ded_{\call{I}} F$.
\end{lemma}  
\begin{proof}
  Let $E$ and $F$ be two sets of ground terms in normal form and assume there is
  a rule $x_1,\ldots,x_n\ded f(x_1,\ldots,x_n) \in\call{L}_0$ such that
  $E\ded_{x_1,\ldots,x_n\ded f(x_1,\ldots,x_n)} F$. By definition there exists a
  ground substitution $\sigma$ in normal form such that
  $(x_1,\ldots,x_n)\sigma\subseteq E$ and
  $F=E\cup\set{\norm{f(x_1,\ldots,x_n)\sigma}}$.  Due to the finite variant
  property, there exists a variant substitution $\theta$ of $f(x_1,\ldots,x_n)$
  and a ground normal substitution $\sigma'$ such that
  $\norm{f(x_1,\ldots,x_n)\sigma}=\norm{f(x_1,\ldots,x_n)\theta}\sigma'$ and
  $\sigma=\theta\sigma'$. The rule $Img(\theta)\ded
  \norm{f(x_1,\ldots,x_n)\theta}$ was added to \call{L} by Step 1 this implies
  that $E\ded _{\call{I}} F$.  To prove the converse, notice that if
  $(x_1,\ldots,x_n)\theta\ded \norm{f(x_1,\ldots,x_n)\theta}$ can be applied
  with the normal ground substitution $\sigma'$ on $E$, then the rule
  $x_1,\ldots,x_n\ded f(x_1,\ldots,x_n)$ can be applied with the ground
  substitution $\sigma=\norm{\theta\sigma'}$ on $E$.
\end{proof}

Also, the computation of Step.~2 is correct and complete in the following
sense.

\begin{lemma}{\label{lemma:00:lem1}}
  For any set of ground terms $E$ in normal form and any ground term $t$ in
  normal form we have: $t\in\rhclos{E}{\call{I}}$ if and only if
  $t\in\rhclos{E}{\II}$.
 \end{lemma}

\begin{proof}
  The direct implication is trivial since $\call{L}'$ is initialised with
  $\call{L}$.  Let us prove the converse implication. Assume that there exists a
  $\II$-derivation starting from $E$ of goal $t$. Let us define an arbitrary
  total order on the rules of $\call{L}$, and we extend this order to rules of
  $\call{L}'\setminus\call{L}$ as follows: rules of $\call{L}$ are smaller than
  the rules of $\call{L}'\setminus\call{L}$ and rules of
  $\call{L}'\setminus\call{L}$ are ordered according to the order of their
  construction during the saturation.  Let $\M{D}$ be the multiset of rules
  applied in $D$.  Let $\Omega(E,t)=\set{D\mid D:E\ded^*_{\II}F\ni t}$. By
  construction, the ordering on rules is total and well-founded, and thus the
  pre-ordering on derivations in $\Omega(E,t)$ is also total and well-founded.
  Since $t\in\rhclos{E}{\II}$, we have $\Omega(E,t)\not=\emptyset$, and thus
  $\M{\Omega(E,t)}$ has a minimum element which is reached. Let $D$ be a
  derivation in $\Omega(E,t)$ having the minimum $\M{D}$, and let us prove that
  $D$ employs only rules in $\call{L}$.  By contradiction, assume that $D$ uses
  a rule $l\ded r \in\call{L}'\setminus\call{L}$ applied with a ground
  substitution $\sigma$ on a set $F$.  Since $l\ded r \notin\call{L}$, it has
  been constructed by closure rule. Thus, there exists two rules $l_1\ded
  r_1\in\call{L}'_{inc}$ and $l_2\ded r_2\in\call{L}'$, a term $s\in
  l_2\setminus\Variables$ such that $s$ and $r_1$ are unifiable,
  $\alpha=mgu(s,r_1)$, $l=(l_1,l_2\setminus s)\alpha$ and $r=r_2\alpha$.
  Replacing the application of the rule $l\ded r$ by two steps applying first
  the rule $l_1\ded r_1$ and then $l_2\ded r_2$ yields another derivation $D'$.
  Since $l\ded r$ must have an order bigger than the order of $l_1\ded r_1$ and
  $l_2\ded r_2$ and the last two rules are in $\call{L}'$, we deduce that
  $D'\in\Omega(E,t)$ and $\M{D'}<\M{D}$ which contradicts the minimality of
  $\M{D}$.
\end{proof}

 Let $E$ (resp. $t$) be a set of terms (resp. a term) in normal form and let $D$
 be a derivation starting from $E$ of goal $t$, $D:E=E_0\ded E_0,t_1\ded
 \ldots\ded E_{n-2},t_{n-1}\ded E_{n-1},t$.  The derivation $D$ is
 \emph{well-formed} if for all rules $l\ded r$ applied with substitution
 $\sigma$, for all $u\in l\setminus\Variables$ we have either $u\sigma\in E$ or
 $u\sigma$ was deduced  by a former decreasing rule. The following lemma is a
 consequence of the computation of the closure. Notice that we do not assume
 here, nor afterward unless  stated, that the saturation terminates.

\begin{lemma}{\label{lemma:00:saturation2}}
  Let $E$ (resp. $t$) be a set of terms (resp. a term) in normal form such that
  $t\in \rhclos{E}{\II}$. For all $\call{I}'$-derivations $D$ starting from $E$
  of goal $t$ we have either $D$ is well-formed or there is another
  $\call{I}'$-derivation $D'$ starting from $E$ of goal $t$ such that
  $\Cons{D}=\Cons{D'}$ and $D'$ is well-formed.
 \end{lemma}
 \begin{proof}
   We have $t\in\rhclos{E}{\II}$ implies that the set $\Omega(E,t)$ of
   \II-derivations starting from $E$ of goal $t$ is not empty.  Let
   $D\in\Omega(E,t)$, $D:E=E_0\ded E_1\ded \ldots \ded E_{n-1},t$, we denote
   $l_i\ded r_i$ the rule applied at step $i$ with the substitution $\sigma_i$
   and suppose that $D$ is not well-formed. Let us (pre-)order derivations in
   $\Omega(E,t)$ with a measure $M$ such that $\M{D'}$ for a derivation $D'$ is a
   multiset of integers constructed as follows: starting with $\M{D'}=\emptyset$,
   for all steps $k$, $1\leq k\le n$, for all terms $u\in l_k\sigma_k$ obtained
   by former increasing rule, add $k$ to $\M{D'}$.  Since this pre-order is
   well-founded, there exists a derivation $d\in\Omega(E,t)$ such that $\M{d}$
   is minimum and $\Cons{d}=\Cons{D}$. Let us prove that $d$ is well-formed.  By
   contradiction, assume that $d$ is not well-formed and let $j$ be the first
   step in $d$ such that $l_j\ded r_j$ is the rule applied with substitution
   $\sigma_j$ and there is a term $u\in l_j\setminus\Variables$ obtained by a
   former increasing rule, let $l_h\ded r_h$ be this rule.  Since $l_h\ded
   r_h\in\call{L}'_{inc}$ and $u\notin\Variables$, \emph{Closure} can be applied
   on $l_h\ded r_h$ and $l_j\ded r_j$ and the resulting rule can be applied at
   step $j$ instead of $l_j\ded r_j$ yielding also $E_{j}$. Let $d'$ be the
   derivation obtained after this replacement, $d'\in\Omega(E,t)$ and
   $\Cons{d'}=\Cons{d}$.  Since $h<j$ and by definition of M, we have
   $\M{d'}<\M{d}$ which contradicts the minimality of $\M{d}$. We deduce that
   $d$ is well-formed and then we have the lemma.
 
\end{proof}

\section{Reachability problems}\label{sec:decision}

\subsection{Presentation of the algorithm and pre-computation}

This section is devoted to the presentation of an algorithm for solving
Reachability Problems and to a proof scheme of its completeness, correctness and
termination. In this section, we denote by
$\call{I}_0=\intrus{\call{G}}{\call{L}_0}{\call{H}}$ the initial deduction
system and by $\call{I}'=\intrus{\call{G}}{\call{L}'}{\emptyset}$ the saturated
deduction system. From now, we suppose that $\call{L}'$ is finite and we recall
that $\call{L}'$ is partitioned into two disjoint sets of deduction rules
$\call{L}'_{inc}$ and $\call{L}'_{dec}$ (by definition of \emph{increasing} and
\emph{decreasing} rules).  The algorithm comprises two steps, and is depicted in
Fig.~\ref{fig:algo:solve}

\begin{figure}[htbp]
  \centering
  \parbox{0.9\linewidth}{\underline{Resolution($\call{C}^0$)}\\
    We let $\call{C}^0=((E_i^0\rhd v_i^0)_{i\in\set{1,\ldots,n}},\call{S}^0)$ be an $\call{I}_0$-constraint system.
  \begin{description}
  \item[Step~1.] Guess a finite variant substitution $\theta$ for all terms of
    $\call{C}^0$, apply $\theta$ on these terms  and normalise them then solve the obtained unification system.  Finally, apply the obtained solution $\alpha$ on the constraints.  In the sequel we will abuse notations and denote the obtained
    constraint system $\call{C}=(E_i\rhd t_i)_{i\in\set{1,\ldots,n}}$, where
    $t_i=\norm{v_i^0\theta}\alpha$ and $E_i=\norm{E_i^0\theta}\alpha$.
  \item[Step~2.] Apply non-deterministically the transformation rules of
    Fig.~\ref{fig:resolution}
  \item[Step 3.] If a solved form is reached, return \textsc{Sat}, else return
    \textsc{Fail}.
  \end{description}}
  \caption{Algorithm for solving constraint systems.}
  \label{fig:algo:solve}
\end{figure}

\paragraph{Remarks.} 
\begin{description}
\item[Solved form.]  A constraint system \call{C} as denoted at the end of the
  first step is in \emph{solved form} if for all constraints
  $E\rhd{}t\in\call{C}$ we have $t\in\Variables$.  Every constraint system in
  solved form has at least one solution~\cite{AmadioL00}.
\item[Computation of the finite variants substitutions.] Given
  $\call{C}^0=((E_i^0\rhd{}v_i^0)_{1\le i\le n},\call{S}^0)$, and let $T$ be a
  n-uplet containing terms appearing in $\call{C}^0$, $T=\langle
  u_1,\ldots,u_n\rangle$. Due to the finite variant property, $T$ has finite set
  of variant substitutions.  We choose a variant substitution $\theta$ among the
  possible ones.
\item[Justification of the first step.] Let $\sigma$ be a normal solution of the
  original constraint system. The first step will non-deterministically
  transform terms of \call{C}, $u_1,\ldots,u_n$, into terms $u_1^0,\ldots,u_n^0$
  such that, according to definition \ref{def:FVP} we will have
  $\set{\norm{u_i\sigma}=u_i^0\sigma'}_{1\leq i\leq n}$ for a normal
  substitution $\sigma'$. It is easily verified that the first step always
  terminates.
\end{description}

We prove below that there exists a solution to the original
$\call{I}_0$-constraint system $\call{C}^0$ \textit{iff} there exists a solution
to one of the possible constraint systems computed in the first step for the
\II{} deduction system.

\begin{lemma}{\label{lemma:00:lem2}(Completeness)}
  Let $\call{C}^0$ be an $\call{I}_0$-constraint system.  If $\call{C}^0$ is
  $\call{I}_0$-satisfiable, there exists a constraint system $\call{C}$ in the
  output of Step~1. such that $\call{C}$ is $\II$-satisfiable.
\end{lemma}

\begin{proof} 
  We have
  $\call{C}^0=((E_i^0\rhd{}v_i^0)_{i\in\set{1,\ldots,n}},\call{S}^0)$. Let
  $\sigma$ be a substitution in normal form such that
  $\sigma\models_{\call{I}_0}\call{C}^0$.  This implies that
  $\norm{v_i^0\sigma}\in\rhclos{\norm{E_i^0\sigma}}{\call{I}_0}$ for
  $i\in\set{1,\ldots,n}$ and thus, by lemmas \ref{lemma:00:saturation1} and
  \ref{lemma:00:lem1}, $\norm{v_i^0\sigma}\in\rhclos{\norm{E_i^0\sigma}}{\II}$
  for $i\in\set{1,\ldots,n}$. We have also $\norm{s^0\sigma}=\norm{s'^0\sigma}$
  for all equations $s^0\unif{}s'^0\in\call{S}^0$.  By definition \ref{def:FVP},
  there exists a variant substitution $\theta$ of the terms in $\call{C}^0$ and
  a substitution $\sigma'$ in normal form such that for each term
  $u\in\call{C}$, we have $\norm{u\sigma}=\norm{u\theta}\sigma'$.  This implies
  that $\norm{v_i^0\theta}\sigma'\in\rhclos{\norm{E_i^0\theta}\sigma'}{\II}$ for
  $i\in\set{1,\ldots,n}$ and $\norm{s^0\theta}\sigma'=\norm{s'^0\theta}\sigma'$
  for all equations $s^0\unif{}s'^0\in\call{S}^0$.  The unification system
  $\norm{\call{S}^0\theta}$ has solution $(\sigma')$, let $\mu$ be its most general unifier,
  we have $\sigma'=\mu\alpha$ for some substitution $\alpha$ and
  $\alpha\models_{\II} \norm{E_i^0\theta}\mu\rhd{}\norm{v_i^0\theta}\mu$ for
  $i\in\set{1,\ldots,n}$.  The constraint system
  $\call{C}=((\norm{E_i^0\theta}\mu\rhd{}\norm{v_i^0\theta}\mu)_{i\in\set{1,\ldots,n}})$ is a possible
  output of Step~1 and it is $\II$-satisfiable.
\end{proof}

\begin{lemma}{\label{lemma:00:lem3}(Correctness)}
  Let $\call{C}^0$ (resp.  \call{C}) be a $\call{I}_0$- (resp. \II-) constraint
  system. Assume that $\call{C}$ is obtained from $\call{C}^0$ by applying
  Step~1.  If \call{C} is satisfiable then so is $\call{C}^0$.
\end{lemma}

 \begin{proof}
   Let $\call{C}^0$ (resp. \call{C}) be a $\call{I}_0$- (resp. \II-) constraint
   system and assume that $\call{C}$ is obtained from $\call{C}^0$ by applying
   Step~1.  This implies that
   $\call{C}^0=((E_i^0\rhd{}v_i^0)_{i\in\set{1,\ldots,n}},\call{S}^0)$ and
   $\call{C}=((\norm{E^0_i\theta}\mu\rhd{}\norm{v_i^0\theta}\mu)_{i\in\set{1,\ldots,n}})$
   while $\theta$ is a variant substitution of the terms of $\call{C}^0$ and
   $\mu$ is the most general unifier of the unification system
   $\norm{\call{S}^0\theta}$ obtained from $\call{S}^0$ by applying the variant
   substitution $\theta$ on the terms of $\call{S}^0$ and then normalising these
   terms.  Since $\call{C}$ is $\II$-satisfiable there exists a normal
   substitution $\sigma$ such that
   $\norm{v_i^0\theta}\mu\sigma\in\rhclos{\norm{E_i^0\theta}\mu\sigma}{\II}$ and
   thus
   $\norm{v_i^0\theta\mu\sigma}\in\rhclos{\norm{E_i^0\theta\mu\sigma}}{\call{I}_0}$ (lemmas \ref{lemma:00:saturation1} and \ref{lemma:00:lem1}). We
   conclude that $\norm{\theta\mu\sigma}\models_{\call{I}_0}\call{C}^0$.
\end{proof}

\subsection{Transformation in solved form}

Let $\II=\intrus{\call{G}}{\call{L}'}{\emptyset}$ be the deduction system
resulting from the application of the saturation algorithm.  In the rest of this
paper, we denote by $l_x,l_1,\ldots,l_n\ded r$ a $\call{L}'$-rule such that
$l_x$ is a finite set of variables and $\set{l_1,\ldots,l_n}$ is a finite set of
non-variable terms. Unless otherwise specified, \II{} is the deduction system
implicit in all notations.

In the rest of this section, we prove a \emph{progress} property: If a
satisfiable constraint system is not in solved form, then a rule of
Fig.~\ref{fig:resolution} can be applied on it to yield another satisfiable
constraint system. We will give conditions in the next section ensuring the
termination of the application of these rules.

\begin{figure}[htbp]
  $$
  \begin{array}{lc}
    
    \mathsf{Unif}:&
    \hspace*{-2.5em}
    \vcenter{
      \infer
      [
      \begin{array}{c}
        u\in E\setminus\Variables, ~  t\notin\Variables,\\
        \sigma= mgu{}(u,t)
      \end{array}
      ]%
      {
        (\call{C}_\alpha,\call{C}_\beta)\sigma
      }
      {
        \call{C}_\alpha,E\rhd{}t,\call{C}_\beta
      }
    }\\[1em]
    
    \mathsf{Reduce~1}:&\\
    \multicolumn{2}{c}{%
      \infer
      [\begin{array}{c}
        l_x,l_1,\ldots,l_n\ded r\in\call{L}'_{inc}\text{ and }
        t\notin\Variables\\
        e_1,\ldots,e_n\in E\setminus\Variables \text{ and }
        \sigma=mgu(\set{e_i\unif{}l_i}_{1\leq i\leq n},r\unif{}t)\\
      \end{array}]
      {
        (\call{C}_\alpha, (E\rhd{} y)_{y\in l_x},  \call{C}_\beta)\sigma
      }
      {
        \call{C}_\alpha,
        E\rhd{}t,
        \call{C}_\beta
      }
    }\\[1em]
    
    \mathsf{Reduce~2}:&\\
    \multicolumn{2}{c}{%
      \infer
      [\begin{array}{c}
        l_x,l_1,\ldots,l_n\ded r\in\call{L}'_{dec}\text{ and }
        t\notin\Variables\\
        e_1,\ldots,e_n\in E\setminus\Variables\text{ and }
        \sigma=mgu(\set{e_i\unif{}l_i}_{1\leq i\leq n})\\
        \call{C}'_\beta ~ is ~ obtained ~ from ~ \call{C}_\beta~by \\
        adding~r~to~left~hand~side~ of~constraints\\
      \end{array}]
      {
        (\call{C}_\alpha, (E\rhd{} y)_{y\in l_x},  E\cup r\rhd{} t ,\call{C}'_\beta)\sigma
      }
      {
        \call{C}_\alpha,
        E\rhd{} t,
        \call{C}_\beta
      }
    } 
  \end{array}
  $$
  \vspace*{-1em}
  \caption{\label{fig:resolution} System of transformation rules.}
  \vspace*{-2em}
\end{figure}
 
The progress proof relies on two normalisation lemmas for constraint systems.

\begin{lemma}{\label{lemma:00:eliminationvariables}}
  Let $\call{C}=(\call{C}_\alpha,E\rhd{}t,\call{C}_\beta)$ be a constraint
  system such that $\call{C}_\alpha$ is in solved form.  Then, for all
  substitutions $\sigma$ we have: $\sigma\models\call{C}$ if and only if
  $\sigma\models (\call{C}_\alpha,(E\setminus
  \Variables)\rhd{}t,\call{C}_\beta)$ .
\end{lemma}

\begin{proof} It suffices to prove that if $x\in E\cap\Variables$ and $\sigma$
  is a substitution such that $\sigma\models\call{C}$, then we have
  $\sigma\models(\call{C}_\alpha,(E\setminus \set{x})\rhd{}t,\call{C}_\beta)$.
  Given $x\in E$, by definition \ref{def:constraints}, there exists a set of
  terms $E_x\subseteq E$ such that $E_x\rhd{}x\in\call{C}_\alpha$.  Since
  $\sigma\models\call{C}$ we have $\sigma\models E_x\rhd{}x$, and by the fact
  that $E_x\subseteq E\setminus\set{x}$ we have $\sigma\models
  E\setminus\set{x}\rhd{}x$.  Since we also have $\sigma\models (E\rhd{}t) $
  then, $\sigma\models E\setminus\set{x}\rhd{}t$.  The reciprocal is obvious
  since $E\setminus\Variables\subseteq E$.
\end{proof}

\begin{lemma}{\label{lem:elimination}}
  Let $\call{C}=(\call{C}_\alpha,E\rhd{}x,\call{C}_\beta)$ be a constraint
  system such that $\call{C}_\alpha$ is in solved form and
  $x\notin\Var{\call{C}_\alpha,E,\call{C}_\beta}$ and let
  $\call{C}'=(\call{C}_\alpha,\call{C}_\beta)$. We have:
\begin{enumerate}
\item
If  $\sigma\models\call{C}$  then  $\sigma\models \call{C}'$.
\item
If $\sigma'\models\call{C}'$ then we can extend $\sigma'$ to $\sigma$ 
 such that  $\sigma\models\call{C}$.
\end{enumerate}
\end{lemma}

\begin{proof}
\begin{enumerate}
\item Let $\call{C}=(\call{C}_\alpha,E\rhd{}x,\call{C}_\beta)$ and let $\sigma$
  be a closed substitution such that $\sigma\models\call{C}$.  Since
  $x\notin\Var{\call{C}_\alpha,E,\call{C}_\beta}$, we deduce that
  $\call{C}'=(\call{C}_\alpha,\call{C}_\beta)$ is deterministic and
  $\sigma\models\call{C}'$.
\item Let $\sigma'$ be a closed substitution such that
  $\sigma'\models\call{C}'$. Since $\Var{E}\subseteq\Var{\call{C}_\alpha}$,
  $\sigma'$ is defined on $\Var{\call{C}_\alpha,E,\call{C}_\beta}$ and since
  $x\notin\Var{\call{C}_\alpha,\call{C}_\beta}$, $\sigma'(x)$ is not defined. We
  extend $\sigma'$ to $\sigma$ as follows:

  $\sigma(y)=\sigma'(y)$ for $y\in\Supp{\sigma'}$, $\sigma(x)$ is a closed term
  in E.

  Since $x\notin\Var{\call{C}_\alpha,\call{C}_\beta,E}$ and $x\sigma\in
  E\sigma$, we deduce that $\sigma\models\call{C}$.
\end{enumerate}
\end{proof}

\paragraph{Simplification step.} 
Let $\call{C}=(\call{C}_\alpha,E\rhd{}t,\call{C}_\beta)$ be a constraint system
such that $\call{C}_\alpha$ in solved form and $t\notin\Variables$.  If we apply
\emph{Reduce 1} (resp. \emph{Reduce 2}) on \call{C} using a rule
$l_x,l_1,\ldots,l_n\ded r$ such that there is a variable $x\in
l_x\setminus\Var{l_1,\ldots,l_n,r}$ then the constraint $E\rhd{}x$ will be in
the obtained constraint system $\call{C}'$ and $x$ does not appear twice in
$\call{C}'$. By lemma \ref{lem:elimination}, this constraint can be deleted from
$\call{C}'$.  As a consequence, we apply a simplification step on the saturated
deduction system $\call{L}'$ that eliminates variables $x\in
l_x\setminus\Var{l_1,\ldots,l_n,r}$ for all rules $l_x,l_1,\ldots,l_n\ded
r\in\call{L}'$.

Each of the rules in Fig.~\ref{fig:resolution} is correct and complete
w.r.t. the satisfiability of constraint systems.

\begin{lemma}{\label{lemma:00:completeness}}
  A satisfiable constraint system not in solved form can be reduced into another
  satisfiable constraint system by applying a rule of figure \ref{fig:resolution}.
\end{lemma}

\begin{proof}
  Let $\call{C}=(E_j\rhd{}t_j)_{1\leq j\leq n}$ be a satisfiable constraint
  system not in solved form and let $i$ be the smallest integer such that
  $t_i\notin\Variables$. Let
  $\call{C}=(\call{C}_\alpha,E_i\rhd{}t_i,\call{C}_\beta)$ where
  $\call{C}_\alpha$ is in solved form. Since $\call{C}$ is satisfiable there
  exists a substitution $\sigma$ such that
  $\sigma\models_{\call{I}'}\call{C}$. Let us prove that $\call{C}$ can be
  reduced into another satisfiable constraint system $\call{C}'$ by applying
  transformation rules given in figure \ref{fig:resolution}. By lemma
  \ref{lemma:00:eliminationvariables}, $\sigma\models_{\mathcal{I}'}\call{C}$
  implies $\sigma\models_{\mathcal{I}'}(\call{C}_\alpha,E_i \setminus \Variables
  \rhd{} t_i, \call{C}_\beta)$ and that, by lemma \ref{lemma:00:saturation2},
  there is a well-formed derivation $D$ starting from
  $(E_i\setminus\Variables)\sigma$ of goal $t_i\sigma$.  We have two cases:
  \begin{itemize}
  \item If $t_i\sigma\in (E_i\setminus\Variables)\sigma$ then there exists a
    term $u\in E_i\setminus\Variables$ such that $u\sigma=t_i\sigma$. Let
    $\mu=mgu(t_i,u)$, we have $\sigma=\mu\theta$ for some substitution $\theta$.
    $\call{C}$ can then be reduced to $\call{C}'$ by applying \emph{Unif} rule,
    $\call{C}'=(\call{C}_\alpha\mu, \call{C}_\beta\mu)$ and
    $\theta\models_{\mathcal{I}'}\call{C}'$.
  
\item If $t_i\sigma \notin (E_i\setminus\Variables)\sigma$, let $D: ~
  (E_i\setminus\Variables)\sigma\ded \ldots\ded F\sigma,t_i\sigma$ and for every
  step in $D$ where $l\ded r$ is the rule applied with the substitution
  $\gamma$, for every $s\in l\setminus\Variables$, we have either $s\gamma\in
  (E_i\setminus\Variables)\sigma$ or $s\gamma$ was constructed by a former
  decreasing rule.
  \begin{itemize}
  \item Suppose that all applied rules in $D$ are increasing and let $l\ded r$
    be the last applied rule with the substitution $\gamma$, this implies that
    $r\gamma=t_i\sigma$ and for every $s\in l\setminus\Variables$, $s\gamma\in
    (E_i\setminus\Variables)\sigma$ and then for every $s\in
    l\setminus\Variables$ there exists a term $u\in E_i\setminus\Variables$ such
    that $s\gamma=u\sigma$.  Let $\mu$ be the most general unifier of
    $\set{r\unif{}t_i,(s\unif{}u)_{\forall s\in l\setminus\Variables,u\in
        E_i\setminus\Variables ~and~ s\gamma=u\sigma}}$, we have
    $\sigma=\mu\theta$ and $\gamma=\mu\theta$ for some $\theta$.  This implies
    that $\call{C}$ can be reduced to $\call{C}'=(\call{C}_\alpha,
    (E_i\rhd{}x)_{x\in l}, \call{C}_\beta)\mu$ by applying \emph{Reduce 1} and
    $\theta\models_{\mathcal{I}'}\call{C}'$.
  \item Suppose that $D$ contains decreasing rules and let $j$ be the first step
    where the applied rule is decreasing. Let $l\ded r$ be this rule applied
    with substitution $\gamma$.  $D: ~
    (E_i\setminus\Variables)\sigma=F_0\sigma\ded
    F_0\sigma,t_1\sigma\ded\ldots\ded F_{j-1}\sigma\ded
    F_{j-1}\sigma,t_{j}\sigma\ded\ldots\ded F_{n-1}\sigma,t_i\sigma$. Since $D$
    is well-formed, we deduce that for every $s\in l\setminus\Variables$,
    $s\gamma\in (E_i\setminus\Variables)\sigma$ and then, for every $s\in
    l\setminus\Variables$ there exists a term $u\in E_i\setminus\Variables$ such
    that $s\gamma=u\sigma$. Let $\mu$ be the most general unifier, we have
    $\gamma=\mu\theta$ and $\gamma=\mu\theta$ for some substitution
    $\theta$. This implies that $\call{C}$ can be reduced to
    $\call{C}'=(\call{C}_\alpha , (E_i\rhd{} x)_{x\in l}, (E_i\cup r)\rhd{} t_i,
    \call{C}'_\beta)\mu$ by applying \emph{Reduce 2} and
    $\theta\models_{\mathcal{I}'} \call{C}'$.
  \end{itemize} 
    
  \end{itemize}
\end{proof}

\begin{lemma}{\label{lemma:00:correcteness}}
  Let $\call{C}$ and $\call{C}'$ be two constraint systems such that $\call{C}'$
  is obtained from $\call{C}$ by applying a transformation rule. If \call{C}' is
  satisfiable then so is \call{C}.
 \end{lemma}

\begin{proof}
  Let $\call{C}$ and $\call{C}'$ be two constraint systems such that $\call{C}'$
  is obtained from $\call{C}$ by applying a transformation rule and suppose that
  $\call{C}'$ is satisfiable. Let $\sigma'$ be a solution of $\call{C}'$ and let
  us prove that $\call{C}$ is satisfiable.  Since a transformation rule can be
  applied on $\call{C}$, $\call{C}$ can't be in solved form. Suppose that
  $\call{C}=(\call{C}_\alpha,E \rhd{}t, \call{C}_\beta)$ where $\call{C}_\alpha$
  is in solved form and $t\notin\Variables$.
\begin{itemize}
\item If $\call{C}'$ is obtained from $\call{C}$ by applying \emph{Unif} rule
  then, there exists a term $u\in E\setminus\Variables$ such that $u$ and $t$
  are unifiable.  Let $\mu$ be the most general unifier then
  $\call{C}'=(\call{C}_\alpha\mu, \call{C}_\beta\mu)$.  Since
  $\sigma'\models_{\mathcal{I}'}\call{C}'$, we have
  $\sigma'\circ\mu\models_{\call{I}'}(\call{C}_\alpha,\call{C}_\beta)$ and by
  the fact that $\mu$ is the most general unifier of $t$ and a term in $E$ we
  have $\sigma'\circ\mu\models_{\call{I}'} E\rhd{}t$.  We deduce that
  $\sigma'\circ\mu\models_{\call{I}'} \call{C}$.

\item If $\call{C}'$ is obtained from $\call{C}$ by applying \emph{Reduce 1}
  then there exists an increasing rule $l_x,l_1,\ldots, l_n\ded r$, a set of
  terms $e_1,\ldots,e_n$ in $E\setminus\Variables$ such that $\set{r\unif{}t,
    (l_i\unif{}e_i)_{1\leq i\leq n}}$ has solution. Let $\mu$ be the most
  general unifier. $\call{C}'=(\call{C}_\alpha, (E\rhd{}x)_{x\in l_x},
  \call{C}_\beta)\mu$.  Since $\sigma'\models_{\call{I}'}\call{C}'$ and by
  definition of $\mu$, we have $\sigma'\circ\mu\models_{\call{I}'}\call{C}$.

\item If $\call{C}'$ is obtained from $\call{C}$ by applying \emph{Reduce 2}
  then there exists a decreasing rule $l_x,l_1,\ldots,l_n\ded r$ and a set of
  terms $e_1,\ldots,e_n$ in $E\setminus\Variables$ such that
  $\set{(l_i\unif{}e_i)_{1\leq i\leq n}}$ has solution.  Let $\mu$ be the most
  general unifier. $\call{C}'=(\call{C}_\alpha, (E\rhd{}x)_{x\in l_x}, (E\cup
  r)\rhd{} t, \call{C}'_\beta)\mu$. Since
  $\sigma'\models_{\mathcal{I}'}\call{C}'$ and by definition of $\mu$ and
  constraint systems, we have $\sigma'\circ\mu\models_{\mathcal{I}'}\call{C}$.

\end{itemize}
\end{proof}

\section{Decidability of reachability problems}\label{sec:ground.reachability}

In this section we first prove that if the saturation terminates then ground  reachability problems are decidable. We then give an additional criterion
 that will permit us to lift this result to general reachability problems.

\subsection{Decidability of ground reachability problems}\label{ground}
We recall that $\call{I}_0=\intrus{\call{G}}{\call{L}_0}{\call{H}}$ is the
initial deduction system and $\call{I}'=\intrus{\call{G}}{\call{L}'}{\emptyset}$
is the saturated deduction system.

Let us also first recall in the following lemma some properties of
reduction ordering.

\begin{lemma}{\label{lem:lem12}}
Let $t_1,t_2\in\vsig{\call{G}}$ and
$t_1\preceq t_2$. We have:
\begin{enumerate}
\item
 $\Var{t_1}\subseteq\Var{t_2}$
\item
$t_2\notin\SSub{t_1}$
\item 
If $t_2\in\Variables$ then  $t_1=t_2$
\item
If  $t_1\notin \Variables$ then $t_1\not\prec x$
\end{enumerate}
\end{lemma}
\begin{proof}

\begin{enumerate}
\item Let $t_1$ and $t_2$ be two terms and $t_1\preceq t_2$.  If $t_1=t_2$ then
  we have obviously $\Var{t_1}=\Var{t_2}$.  Suppose $t_1\not=t_2$ this implies
  that $t_1\prec t_2$ and let us prove that $\Var{t_1}\subseteq\Var{t_2}$.  By
  contradiction, suppose that $\Var{t_1}\not\subseteq\Var{t_2}$ and let
  $x\in\Var{t_1}\setminus\Var{t_2}$.  By definition of $\prec$, we have
  $t_1\sigma\prec t_2\sigma$ for all substitutions $\sigma$.  Let $\sigma$ be a
  substitution such that $\Supp{\sigma}=\set{x}$ and $\sigma(x)=t_2$.  This
  implies that $t_2\sigma=t_2$ and $t_2\in\Sub{t_1\sigma}$ which contradicts
  $t_1\prec t_2$.
\item If $t_2\in\SSub{t_1}$ this implies that $t_1\not=t_2$ and $t_2\prec t_1$
  which contradicts $t_1\preceq t_2$.

\item
If $t_2=x$ we deduce that $\Var{t_1}\subseteq\set{x}$
and $x\notin\SSub{t_1}$.
This implies that $t_1=x$.
\item Suppose that $t_1\not=x$ and $t_1\prec x$. This implies that
  $\Var{t_1}\subseteq\set{x}$ and then, either $t_1=x$ or $x\in\SSub{t_1}$. This
  contradicts the fact that $t_1\not=x$ and $x\notin\SSub{t_1}$.
\end{enumerate}
\end{proof}

A core result of this paper is the following lemma. 

\begin{lemma}{\label{lem:terminaison-ID}}
  Let $\call{I}'$ be a saturated deduction system such that $\call{L}'$ is
  finite. Applying the transformation algorithm of Fig.~\ref{fig:resolution} on
  a constraint system $\call{C}$ without instantiating the variables of \call{C}
  yields only a finite number of different constraint systems.
\end{lemma}

\begin{proof}
  Assume the application of rules of Fig.~\ref{fig:resolution} yields an
  infinite sequence of constraint systems $\call{C}_1,\ldots,\call{C}_n,\ldots$.
  Let us prove there is only a finite number of different $\call{C}_i$ when
  identical constraints within a constraint system are identified.

  Let us first prove that there is only a finite number of different left-hand
  side of deduction constraints. The number of different left-hand sides in a
  constraint system does not change (or decrease) when a \textsc{Unif} or
  \textsc{Reduce1} rule is applied.  Assume now that a decreasing rule
  $l_x,l_1,\ldots,l_n\ded r\in\call{L}'$ is applied with a substitution $\sigma$
  on a constraint with left-hand side $E$. If $r\sigma\in E$, the number of
  different left-hand side does not change. Thus let us assume $r\sigma\notin
  E$, and thus $r\sigma\notin \cup\set{l_1\sigma,\ldots,l_n\sigma}$.  Since $r$
  is smaller or equal to a term of the left-hand side of there rule, we have two
  case:
  \begin{itemize}
  \item Either there exists $i$ with $l_i\sigma \succ r\sigma$, and thus there
    exists $e\in E$ such that $e\succ r\sigma$.
  \item Or $r\in l_x\setminus\Var{l_1,\ldots,l_n}$. Then the obtained constraint
    system contains the deduction constraints $E\rhd r$ and $E\cup\set{r}\rhd t$
    and not other constraint contains $r$. By
    Lemma~\ref{lemma:00:eliminationvariables} the obtained constraint system is
    equivalent to the one in which $E\cup\set{r}\rhd t$ is replaced by $E\rhd
    t$. 
  \end{itemize}
  Let us now consider the set $T$ which is the union of all left-hand side of
  deduction constraints reachable from $E$ by employing a decreasing rule. 
  \begin{itemize}
  \item the root is labelled by $\emptyset$;
  \item the sons of the root are labelled by the terms in a left-hand side $E$;
  \item The sons of the non-root node are defined as follows: assume there
    exists two left-hand sides $E'$ and $E''$ where $E'$ is reachable from $E$,
    and there is a decreasing rule whose application leads to the addition of a
    deduction constraint with left-hand side $E''=E'',t_1$. Let $t_2\in E'$ be
    the term strictly greater than $t_1$. We then set $t_1$ as a son of $t_2$.
  \end{itemize}
  Since $t_2\succ t_1$ there is no cycle, and since we consider sets reachable
  from $E$, the ``is son of'' relation is connected. It thus defines a tree. We
  note that $t_2$ is the instance of a non-variable term $l$ in the left-hand
  side of a decreasing rule. There is only a finite number of such terms. Since
  we consider deductions in the empty theory, for each $l$ there is a unique
  substitution $\sigma$ such that $l\sigma=t_2$. Given the above properties of
  reduction ordering we have $\Var{r}\subseteq\Var{l}$ and thus $t_1=r\sigma$ is
  uniquely determined by the rule applied. Thus, each term $t_2$ has a finite
  number of sons $t_1$.  Along each branch of the tree a node $t$ is strictly
  smaller than its parent. Since $\succ$ is a well-founded ordering, this
  implies that each branch is finite. Thus, by K\"onig's Lemma, this tree is
  finite. We conclude that $T$ itself is finite. Each left-hand side of a
  deduction constraint is a subset of $T$, thus there is only a finite number of
  different left-hand sides.

  When applying \textsc{Reduce 1} or \textsc{Reduce 2} on a constraint
  $E\rhd{}t$, the newly introduced constraints $E\rhd t'$ are such that $t'$ is
  a strict subterm of a term in $E$ or $t$. Let $E'\rhd t'$ be a deduction
  constraint reached from $E\rhd t$. Either $t'$ is a subterm of $t$ or there
  exists $E''$ reachable from $E$ such that $t'$ is a strict subterm of
  $E''$. Since there is only a finite number of different $E''$, there is thus
  only a finite number of possible right-hand side of constraints.

  In conclusion only a finite number of deduction constraints $E'\rhd t'$ can be
  reached from a deduction constraint $E\rhd t$. Thus only a finite number of
  constraint systems can be reached from a given one by applying rules that do
  not instantiate the variables in the constraint system.
\end{proof}

\begin{definition}\label{def:groud.constraint.system}
  An $\call{I}_0$-ground constraint system \call{C} is denoted
  $(E_1\rhd{}t_1,\ldots,E_n\rhd{}t_n)$ and is defined by a sequence of pairs
  $(E_i,t_i)_{i\in\set{1,\ldots,n}}$ such that $E_i$ (resp. $t_i$) is a set of
  ground terms (resp.  ground term) in normal form and $E_i\subseteq E_{i+1}$
  for $i\in\set{1,\ldots,n}$.
\end{definition}

We note that an $\call{I}_0$-ground constraint $E\rhd{}t$ is valid if
$t\in\rhclos{E}{\call{I}_0}$. We now consider the following problem:

\begin{decisionproblem}{$\call{I}_0$-Ground Reachability Problem}
  \entreeu{ An $\call{I}_0$-ground constraint system \call{C}.}
  \sortie{\textsc{Val} iff
    $(t_i\in\rhclos{E_i}{\call{I}_0})_{i\in\set{1,\ldots,n}}$.}
\end{decisionproblem}

We recall that $t\in\rhclos{E}{\call{I}_0}$ iff $t\in\rhclos{E}{\call{I}'}$
while $E$ (resp. $t$) is set of closed terms (resp.  closed term) in normal form
(Lemmas~\ref{lemma:00:saturation1} and~\ref{lemma:00:lem1}).  This implies that
solving $\call{I}_0$-ground reachability problem is reduced to solving
$\call{I}'$-ground reachability problem. It is then routine to see that a ground
constraint system is valid if, and only if, it reduces to an empty sequence of
deduction constraints. Thus by Lemma~\ref{lem:terminaison-ID} we have:

\begin{theo}{\label{theo:ground}}
  If the saturation algorithm terminates on $\call{L}_0$, the
  $\call{I}_0$-ground reachability problem is decidable.
\end{theo}

\section{Termination of Saturation does not imply decidability of general reachability
  problems}

It is well-known how to encode 2-stack automata into deduction systems. However
the saturation will typically not terminate on standard encodings as it will
amount in this case to the pre-computation of all possible executions of the
automaton. We can however adapt the construction so that saturation
terminates. We consider a signature   \call{G} such that, for all symbol
$f\in\call{G}_0$ of arity $n$, there is a deduction rule $x_1,\ldots,x_n\ded
f(x_1,\ldots,x_n)$, and the signature   $\call{G}=\call{G}\cup\set{g}$ with $g$ a
symbol of arity $1$. Let $(Q,Q_I,Q_F,\Sigma,\Pi,\Delta)$ be a finite $2$-stack
automaton, where $Q$ is the finite set of states of the automaton, $Q_I$ and
$Q_F$ its initial and final states, $\Sigma$ denotes the alphabet of the words
read by the automaton, and $\Pi$ denotes the elements in the stacks of the
automaton. We shall encode the emptiness of the language recognised by this
automaton into a general reachability problem. Let us assume there exists:
\begin{itemize}
\item $\bot\in\call{G}_0$ be a constant denoting the empty stack or the empty
  word;
\item one unary symbol $u_\alpha$ for each letter $\alpha\in\Sigma\cup\Pi$;
\item one constant $q\in\call{G}$ for each state in $Q$;
\item one symbol $s\in\call{G}$ of arity $4$ where we intend that:
  \begin{itemize}
  \item the first argument represents the word that remains to be read by the
    automaton;
  \item the second argument represents the current state of the automaton;
  \item the third and fourth arguments represent the two stacks of the
    automaton.
  \end{itemize}
\item one symbol $f$ of arity $2$.
\end{itemize}
We represent a transition from a state $\sigma_1$ to a state $\sigma_2$ with a
symbol $\tau$ of arity $1$ and a rewriting rule
$\tau(g(f(\sigma_1,f(\sigma_2,x))))\rew g(f(\sigma_2,x))$. The rewriting system
has no critical pairs, and thus is confluent. Since every narrowing step
decreases strictly the number of ``$\tau$'' symbols in a term, narrowing
terminated, and thus the equational theory has the finite variant property. At
the end of the first step of the saturation the system will contain the rules
enabling the attacker to build sequences of states, and additional rules
$g(f(\sigma_1,f(\sigma_2,x)))\ded g(f(\sigma_2,x))$ that are decreasing for any
recursive path ordering. Since there is no increasing rule with the symbol $g$
in the right-hand side, we leave to the reader the proof that saturation
terminates, and hence that ground reachability problems are decidable.
  
However, the instance of $x$ in the following reachability problem encodes a
word recognised by the automaton after a run encoded by the instance of $y$:
$$
\emptyset\rhd f(s(x,q_0,\bot,\bot),y), g(f(s(x,q_0,\bot,\bot),y))\rhd
g(s(\bot,q_f,\bot,\bot))
$$
This example proves (with $q_0\in Q_I$ and $q_f\in Q_F$) that the saturation can
terminate and yield a deduction system for which general reachability problems
are not decidable.

The undecidability comes from the fact that one can apply an unbounded number of
decreasing rules on a non-ground terms, and from the ``lack of regularity'' on
the terms obtained.

\section{Decidability of general reachability problems}
 We recall that the initial intruder system is given by $\call{I}_0=\intrus{\call{G}}{\call{L}_0}{\call{H}}$
 while $\call{H}$ is generated by a convergent equational theory and has the finite variant property. We recall also that $\II=\intrus{\call{G}}{\call{L}'}{\emptyset}$ is the saturated intruder system.

We give here a simple criterion that permits to ensure the termination of the
resolution of a constraint problem with a saturated deduction system.  Let $T$
be a set of terms, $T=\set{t_1,\ldots,t_m}$, we let $\Delta(T)$ to be the set of
strict maximal subterms of $T$ and we define:

$$
\delta(T)=\left\lbrace
  \begin{array}{ll}
    +\infty & \text{if }T\subseteq\Variables\\
    |T\setminus\Variables| - |\Var{T\setminus\Variables} \setminus
    (T\cap\Variables)| & \text{otherwise.}  \\
  \end{array}
\right.
$$

Now let us define $\mu(T)$.  We consider the image of the set of terms $T$ by
the rewriting system \call{U} containing rules $f(x_1,\ldots,x_n)\rew
x_1,\ldots,x_n$ for every symbol $f$ in the signature of the deduction system.  We
define:
$$
\mu(T)=\min_{
  \begin{array}{c}
T\sigma\rew^*_\call{U} T'\\
\sigma\text{ mgu of subterms of }T\\
\end{array}
}\delta(T')
$$

We extend $\mu$ to rules as follows. Let $\call{L}'$ be the set of deduction
rules. We recall that $\call{L}'$ is partitioned into two disjoint sets of
deduction rules, the set of increasing rules $\call{L}'_{inc}$ and the set of
decreasing rules $\call{L}'_{dec}$.  For every rule $l\ded r\in\call{L}'$,

$ \mu(l\ded r)=\left\lbrace\begin{array}{ll}
    \mbox {$\mu(\Delta(l\setminus\Variables\cup \set{r}))$ if $l\ded r$ is increasing,}\\
    \mbox{$\mu(\Delta(l\setminus\Variables))$ otherwise.}  \\
\end{array}\right.$

\begin{definition}{\label{def:contracting}(Contracting deduction systems)}
  A saturated deduction system
  $\call{I}'=\intrus{\call{G}}{\call{L}'}{\emptyset}$ is
  \emph{contracting} if for all rules $l\ded r$ in $\call{L}'$
  we have $\mu(l\ded r)>0$.
\end{definition}

\begin{lemma}{\label{lem:mu}}
  Let $S=\set{s_1,\ldots,s_n}$ and $T=\set{t_1,\ldots,t_n}$ be two sets of terms
  and let $\sigma$ be the most general unifier of
  $V=\set{s_1\unif{}t_1,\ldots,s_n\unif{}t_n}$.  If $\mu(T)>0$ then either
  $\nbv{s_1,\ldots,s_n}>\nbv{(s_1,\ldots,s_n,t_1,\ldots,t_n)\sigma}$ or
  $\nbv{s_1,\ldots,s_n}=\nbv{(s_1,\ldots,s_n,t_1,\ldots,t_n)\sigma}$,
  $S=S\sigma$ and for all $x\in\Var{T}$ there is $i\in\set{1,\ldots,n}$ such
  that $\sigma(x)\preceq s_i$.
 \end{lemma}

\begin{proof}

  Let $V=\set{s_1\unif{}t_1,\ldots,s_n\unif{}t_n}$.  In order to solve $V$, we
  apply the first step of the unification algorithm of Martelli-Montanari
  \cite{MartelliM82}. We reduce $V$ to
  $V'=\set{x_1\unif{}u_1,\ldots,x_k\unif{}u_k,x_{k+1}\unif{}u_{k+1},\ldots,x_m\unif{}u_m}$
  such that for every equation $x\unif{}u\in V'$, we have either $x\in\Var{S}$
  and $u\in\Sub{T}$ or $x\in\Var{T}$ and $u\in\Sub{S}$.  We suppose that
  $x_j\in\Var{T}$ for $j\in\set{1,\ldots,k}$.
\begin{itemize}
\item If $k=m$ then we have $x_j\in\Var{T}$ for $j\in\set{1,\ldots,m}$.  We
  suppose that $x_i\not=x_j$ for all $i,j\in\set{1,\ldots,m}$ and
  $i\not=j$. This implies that $S\sigma=S$ and $\Var{T}$ are instantiated by
  subterms of $S$, that is $\Var{T}\sigma$ are smaller or equal than terms in
  $S$.  We conclude also that
  $\nbv{s_1,\ldots,s_n}=\nbv{(s_1,\ldots,s_n,t_1,\ldots,t_n)\sigma}$.
\item If $k\not=m$ assume $\set{u_{k+1},\ldots,u_m}\notin\Var{T}$, we have
  different cases:
  \begin{itemize}
  \item If for all different $i,j\in\set{1,\ldots,m}$ we have $x_i\not=x_j$ then
    $m-k$ variables of $S$, $x_{k+1},\ldots,x_m$, are instantiated by subterms
    of $T$, $u_{k+1},\ldots,u_m$. This implies that when we apply $\sigma$ to
    \call{S}, new variables,
    $\Var{u_{k+1},\ldots,u_m}\setminus\set{x_1,\ldots,x_k}$ will appear in
    $S\sigma$.  There exists a set $T'\not\subseteq\Variables$ such that
    $T\rew^*_{\call{U}}T'$ and $T'=\set{x_1,\ldots,x_k,u_{k+1},\ldots,u_m}$.
    Since $\mu(T)>0$, we have
    $|T'\setminus\Variables|>|\Var{T'\setminus\Variables}
    \setminus(x_1,\ldots,x_k)|$.  This implies that
    $\nbv{s_1,\ldots,s_n}>\nbv{(s_1,\ldots,s_n,t_1,\ldots,t_n)\sigma}$.
  \item If there is different $i,j\in\set{1,\ldots,m}$ such that $x_i=x_j$:
    \begin{itemize}
    \item If $i,j\leq k$ then we have to unify two subterms of $S$. Let $u_i$
      and $u_j$ be these two subterms and $\alpha$ be their most general
      unifier.

      Let us apply $\alpha$ on $V$ and to solve $V$ we have to solve
      $V\alpha=\set{s_1\alpha\unif{}t_1,\ldots,s_n\alpha\unif{}t_n}$.  To solve
      $V\alpha$ we reduce it to another system $V"$ where equations have the
      same form as in $V'$.  We note that $\nbv{T}$ in $V\alpha$ is the same as
      in $V$ and $\nbv{S}$ is reduced.

      By the same reasoning as above, we deduce that
      $\nbv{S}>\nbv{S\sigma,T\sigma}$.

    \item If $i,j>k$ then we have to unify two subterms of $T$.  Let $u_i$ and
      $u_j$ be these two subterms and $\alpha$ be their most general
      unifier. Let us apply $\alpha$ on $V$ and to solve $V$ we have to solve
      $V\alpha=\set{s_1\unif{}t_1\alpha,\ldots,s_n\unif{}t_n\alpha}$ and to
      solve $V\alpha$, we have to reduce it to another system $V"$ where
      equations have the same form as in
      $V'$. $V"=\set{x_1\unif{}u_1,\ldots,x_m\unif{}u_m}$ where
      $x_1\ldots,x_k\in\Var{T\alpha}$ and $x_{k+1},\ldots,x_m\in\Var{S}$.  By
      definition of $\mu$ and by following the same reasoning as above, we
      deduce that:
      \begin{itemize}
      \item If $k=m$ and for all different $i,j\in\set{1,\ldots,m}$ we have
        $x_i\not=x_j$, we deduce that $S=S\sigma$, $\Var{T}\sigma$ are         smaller or equals  than terms in $S$ and then $\nbv{S}=\nbv{S\sigma,T\sigma}$.
      \item If $k=m$ and there is different $i,~ j$ such that $x_i=x_j$ then we
        have to unify two subterms of $S$ and then we conclude that
        $\nbv{S}>\nbv{S\sigma,T\sigma}$.
      \item If $k\not= m$ we deduce that $\nbv{S}>\nbv{S\sigma,T\sigma}$.
      \end{itemize} 
    \end{itemize}
  \end{itemize} 
\end{itemize}
\end{proof}

The definition of $\mu$ is tailored to the proof of the following Lemma.

\paragraph{Remark.} 

Let $T$ be a set of terms and let  $\Sigma (T)=\set{\sigma ~  s.t. ~ \sigma ~ is ~ the ~ most ~ general ~ unifier ~ of ~ some ~ subterms ~ of ~ T}$.
We remark that $\mu(T)$ is defined with respect  to  $T\sigma$ for every $\sigma\in \Sigma$.  It will be  more naturel and more general if $\mu(T)$ is defined with respect to $T$ instead of some instances of $T$.
The so-called  general definition will be defined as follow:
$$
\mu(T)=\min_{
  \begin{array}{c}
T\rew^*_\call{U} T'\\
\end{array}
}\delta(T')
$$
Using the general definition of $\mu$, we remark that $\mu(T)>0$ does not imply  $\mu(T\sigma)>0$ for a set of terms $T$ and a subtitution $\sigma\in\Sigma(T)$.
Let $T=\set{f(x,x),f(x,y),f(y,x)}$ and let $\sigma$ be such that $\sigma(x)=y$.
Using the general definition of $\mu$, we remark that $\mu(T)>0$ and $\mu(T\sigma)=0$.

Unfortunately, the lemma \ref{lem:mu}, used in the proof of termination (lemma \ref{lem:terminaison-NOT-ID}),
becomes false  with the general definition.

\begin{lemma}{\label{lem:terminaison-NOT-ID}}
  Let $\II$ be a saturated contracting deduction system, \call{C} be a
  $\II$-constraint system not in solved form. If a transformation is applied on
  \call{C} to yield a constraint system $\call{C}'$, then either
 the substitution applied does not instantiate the variables of \call{C} and
    $\Var{\call{C}'}\subseteq\Var{\call{C}}$ or  $\nbv{\call{C}'}<\nbv{\call{C}}$.
\end{lemma}

\begin{proof}
  Let $\call{C}$ be a constraint system such that a transformation rule can be
  applied on it.  This implies that $\call{C}$ is not in solved form. Let
  $\call{C}=(\call{C}_\alpha,E\rhd{}t,\call{C}_\beta)$ such that
  $\call{C}_\alpha$ is in solved form and $t\notin\Variables$.  We have three
  cases:
  \begin{itemize}
  \item If we apply \emph{Unif} rule on $\call{C}$ then there exists a term
    $e\in E\setminus\Variables$ such that $e$ and $t$ are unifiable and $\sigma$
    is the most general unifier.  \call{C} is then reduced to
    $\call{C}'=(\call{C}_\alpha,\call{C}_\beta)\sigma$. Since we unify two
    subterms of \call{C} in the empty theory, either $\sigma$ does not
    instantiate the variables of $\call{C}$ and then
    $\call{C}'=(\call{C}_\alpha,\call{C}_\beta)$ (which implies that
    $\Var{\call{C}'}\subseteq\Var{\call{C}}$) or $\sigma$ instantiates the
    variables of $\call{C}$ (and then $\nbv{\call{C}'}<\nbv{\call{C}}$).
  \item Assume we apply \textsc{Reduce 1} on \call{C}.  By definition of
    \textsc{Reduce 1} there exists an increasing rule $l_x,l_1,\ldots,l_n\ded r
    \in \mathcal{L}'$, a set of terms $e_1, \ldots, e_n\in E\setminus\Variables$
    such that $\call{S}=\set{r\unif{}t, (e_i\unif{}l_i)_{1\leq i\leq n}}$ has a
    solution. Let $\sigma$ be its  most general unifier. Either
    $\sigma_{|\Var{\call{C}}}=\mathop{Id}$[$\sigma_{|\Var{\call{C}}}=\mathop{Id}$]
    or not.  Let us examine the two cases.
    \begin{description}
    \item[$\sigma_{|\Var{\call{C}}}=\mathop{Id}$.] In this case, $\call{C}$
      is reduced to $\call{C}'=(\call{C}_\alpha, (E\rhd{}x\sigma)_{x\in
        l_x},\call{C}_\beta)$.  For each $l_i\in\set{l_1,\ldots,l_n}$ we have,
      by definition of $\sigma$, $l_i\sigma=e_i$. Also, we have
      $r\sigma=t$. Thus for each $x\in\Var{l_1,\ldots,l_n,r}$ we have
      $\Var{x\sigma}\subseteq \Var{\call{C}}$.  Since
      $l_x\subseteq\Var{l_1,\ldots,l_n,r}$ we deduce that
      $\Var{\call{C}'}\subseteq\Var{\call{C}}$.
    \item[$\sigma_{|\Var{\call{C}}}\neq\mathop{Id}$.]  In this case $\call{C}$
      is reduced to $\call{C}'=(\call{C}_\alpha, (E\rhd{}x)_{x\in
        l_x},\call{C}_\beta)\sigma$.  Since the $e_i$ and $r$ are not variables,
      we can decompose all equations in \call{S} to obtain a set of equations in
      which each equation has a member in $\Delta(l_1,\ldots,l_n,r)$.  Since the
      deduction system is contracting Lemma~\ref{lem:mu} implies
      $\nbv{e_1,\ldots,e_n,t}>\nbv{e_1\sigma,\ldots,e_n\sigma,t\sigma,l_1\sigma,
        \ldots,l_n\sigma,r\sigma}$.  Since $l_x\subseteq\Var{l_1,\ldots,l_n,r}$
      we deduce that $\nbv{\call{C}}>\nbv{\call{C}'}$.
    \end{description}
  \item Let us finally assume \textsc{Reduce 2} is applied.  First let us prove
    we can assume $l_x\cup\Var{r}\subseteq\Var{\set{l_1,\ldots,l_n}}$. Since the
    rule is decreasing there exists a term $l\in l_x\cup\set{l_1,\ldots,l_n}$
    such that $\Var{r}\subseteq \Var{l}$. Thus it suffices to prove
    $l_x\subseteq\Var{\set{l_1,\ldots,l_n}}$. By definition of the
    \textsc{Reduce 2} rule, the constraint system \call{C} is transformed into
    $$
    \begin{array}{cr}
      &(\call{C}_\alpha,(E\rhd y)_{y\in l_x\setminus\set{x}},E\rhd
      x,E\cup\set{x}\rhd t, \call{C}'_\beta)\sigma\\
      = &\call{C}_\alpha\sigma,(E\sigma\rhd y\sigma)_{y\in l_x\setminus\set{x}},E\sigma\rhd
      x,E\sigma\cup\set{x}\rhd t\sigma, \call{C}'_\beta\sigma\\
      \equiv & \call{C}_\alpha\sigma,(E\sigma\rhd y\sigma)_{y\in l_x\setminus\set{x}},E\sigma\rhd
      x,E\sigma\rhd t\sigma, \call{C}_\beta\sigma\\
      \equiv & \call{C}_\alpha\sigma,(E\sigma\rhd y\sigma)_{y\in
        l_x\setminus\set{x}},E\sigma\rhd  t\sigma, \call{C}_\beta\sigma\\
    \end{array}
    $$
    where the first $\equiv$ is by Lemma~\ref{lemma:00:eliminationvariables},
    and the second one by Lemma~\ref{lem:elimination}. Thus the resulting system
    is equivalent for solutions to one in which
    $l_x\subseteq\Var{l_1,\ldots,l_n}$. We can then apply the same reasoning as
    above. 
  \end{itemize}  
\end{proof}

We may now conclude by applying the previous results and again K\"onig's Lemma.

\begin{theo}\label{theo:general} 
  Let $\call{I}_0=\intrus{\call{G}}{\call{L}_0}{\call{H}}$ be a deduction system such that the saturation of $\call{L}_0$ terminates ,
  and the resulting deduction system is contracting. Then the
  $\call{I}_0$-reachability problem is decidable.
\end{theo}

\begin{proof}
  It suffices to prove that the application of rules of
  Fig.~\ref{fig:resolution} terminates. Assume there exists an \II-constraint
  system \call{C} and an infinite sequence of transformations starting from
  \call{C}. Let $\call{C}_1,\ldots,\call{C}_n,\ldots$ be the resulting sequence
  of constraint systems. By Lemma~\ref{lem:terminaison-NOT-ID}, at each step
  $\nbv{\call{C}_{i}}\ge\nbv{\call{C}_{i+1}}$ and if there is equality, then the
  substitution applied on $\call{C}_i$ is the identity (does not instantiate the variables of \call{C}). Since we must have a
  positive number of variables, there is only a finite number of steps where the
  substitution is not the identity. Let $\call{C}_n$ be the resulting constraint
  system. Since all subsequent transformation do not instantiate the variables
  of $\call{C}_n$ and its successor, the sequence has only a finite number of
  different constraint systems.

  Since $\call{L}'$ is finite, each constraint system has only a finite number
  of successors. Thus by K\"onig Lemma there is only a finite number of
  different constraint systems. 
\end{proof}

\section{Some relevant equational theories}
We give here some examples of well-known equational theories where the saturation applied on the corresponding initial set of  deduction rules terminates.

\subsection{Dolev-Yao theory with explicit destructors}
The Dolev-Yao theory 
with explicit destructors is the classical Dolev-Yao model with explicit destructors such as decryption and projections.
This theory is given by the following set of equations:

$ \call{H}_{DV}=\left\lbrace\begin{array}{ll}
    \mbox {$Dec_s(Enc_s(x,y),y)=x$,}\\
    \mbox {$Enc_s(Dec_s(x,y),y)=x$,}  \\
    \mbox {$Dec_a(Enc_a(x,PK(y)),SK(y))=x$,}\\
     \mbox {$Enc_a(Dec_a(x,SK(y)),PK(y))=x$,}\\
     \mbox {$\pi_1(\langle x,y\rangle)=x,$}\\
    \mbox {$\pi_2(\langle x, y \rangle)=y.$}\\  
\end{array}\right.$

By orienting  equations of $\call{H}_{DV}$ from left to right, we obtain a rewrite system $\call{R}_{DV}$ generating $\call{H}_{DV}$. We remark that $\call{R}_{DV}$ is convergent and $\call{H}_{DV}$ has finite variant property.

The initial set of deduction rules  is given by the following set of rules:

$\call{L}_0=\left\lbrace\begin{array}{ll}
\mbox{$x,y\ded \langle x, y \rangle$,}\\
\mbox{$x\ded \pi_1(x)$,}\\
\mbox{$x\ded \pi_2(x)$,}\\
\mbox{$x,y\ded Enc_a(x,y),$}\\
\mbox{$x,y\ded Dec_a(x,y),$}\\
\mbox{$x,y\ded Enc_s(x,y),$}\\
\mbox{$x,y\ded Dec_s(x,y).$}
\end{array}\right.$

The saturatation (modulo the simplification introduced after the lemma \ref{lem:elimination}) outputs the following set of deduction rules:

$\call{L}'=\call{L}_0\cup \left\lbrace\begin{array}{ll}
\mbox{$\langle x, y\rangle\ded x,$}\\
\mbox{$\langle x, y\rangle\ded y,$}\\
\mbox{$Dec_a(x,SK(y)),PK(y)\ded x,$}\\
\mbox{$Enc_a(x,PK(y)),SK(y)\ded x,$}\\
\mbox{$Des_s(x,y),y\ded x,$}\\
\mbox{$Enc_s(x,y),y\ded x,$}\\
\mbox{$x,PK(y),SK(y)\ded x.$}\\
\end{array}\right.$

\subsection{Digital signature theory with duplicate signature key selection  property}

The theory of digital signature with duplicate signature  key selection  property is defined in \cite{ChevalierK07} and is given by the following set of equations:

$ \call{H}_{DSKS}=\left\lbrace\begin{array}{ll}
    \mbox {$Ver(x,Sig(x,SK(y)),PK(y))=1$,}\\
    \mbox {$Ver(x,Sig(x,SK'(y_1,y_2)),PK'(y_1,y_2))=1$,}\\
    \mbox {$Sig(x,SK'(PK(y),Sig(x,SK(y))))=Sig(x,SK(y))$.}\\
 \end{array}\right.$

The equational theory $\call{H}_{DSKS}$ is generated by:

$ \call{R}_{DSKS}=\left\lbrace\begin{array}{ll}
    \mbox {$Ver(x,Sig(x,SK(y)),PK(y)) \rightarrow 1$,}\\
    \mbox {$Ver(x,Sig(x,SK'(y_1,y_2)),PK'(y_1,y_2)) \rightarrow 1$,}\\
    \mbox {$Ver(x,Sig(x,SK(y)), PK'(PK(y),Sig(x,SK(y)))) \rightarrow 1$,}\\
    \mbox {$Sig(x,SK'(PK(y),Sig(x,SK(y)))) \rightarrow Sig(x,SK(y))$.}\\
 \end{array}\right.$

 We remark that $\call{R}_{DSKS}$ is convergent and $\call{H}_{DSKS}$ has the finite variant property.

The initial set of deduction rules  is given by the following set of rules:

$ 
\call{L}_0=\left\lbrace\begin{array}{ll}
\mbox{$x,y\ded Sig(x,y),$}\\
\mbox{$x,y,z\ded Ver(x,y,z),$}\\
\mbox{$x,y\ded SK'(x,y),$}\\
\mbox{$x,y\ded PK'(x,y),$}\\
\mbox{$\emptyset\ded 0,$}\\
\mbox{$\emptyset\ded 1.$}\\
\end{array}\right.
$

The saturatation (modulo the simplification introduced after the lemma \ref{lem:elimination}) outputs the following set of deduction rules:

$\call{L}'=\call{L}_0\cup\left\lbrace\begin{array}{ll} 
\mbox{$x,Sig(x,SK(y)),Pk(y)\ded 1$,}\\
\mbox{$x,Sig(x,SK'(y_1,y_2)),PK'(y_1,y_2)\ded 1$,}\\
\mbox{$x,Sig(x,SK(y)),PK'(PK(y),Sig(x,SK(y))) \ded 1$,}\\
\mbox{$x,SK'(PK(y),Sig(x,SK(y)))\ded Sig(x,SK(y))$,}\\ 
\mbox{$SK(y), PK(y)\ded 1$,}\\
\mbox{$SK'(y_1,y_2),PK'(y_1,y_2)\ded 1$,}\\
\mbox{$x,SK(y),PK'(PK(y),Sig(x,SK(y)))\ded 1$,}\\
\mbox{$x, PK(y),Sig(x,SK(y))\ded Sig(x,SK(y))$,}\\
\mbox{$x,PK(y),SK(y)\ded Sig(x,SK(y))$,}\\
\mbox{$y_1,y_2,PK'(y_1,y_2)\ded 1$,}\\ 
\mbox{$x,y_1,y_2,Sig(x,SK'(y_1,y_2))\ded 1$,}\\
\mbox{$y_1,y_2,SK'(y_1,y_2)\ded 1$,}\\
\mbox{$x,PK(y),Sig(x,SK(y))\ded 1$,}\\
\mbox{$x,PK(y),SK(y),Sig(x,SK(y))\ded 1$,}\\
\mbox{$x,SK(y),PK(y),PK'(PK(y),Sig(x,SK(y)))\ded 1$,}\\
\mbox{$x,SK(y),PK(y)\ded Sig(x,SK(y))$.}\\
\end{array}\right.$

\section{Decidability of ground reachability problems  for the blind signature theory}\label{sec:BS}
%
%

Blind signature was introduced  in \cite{KremerR05}, it is defined by the signature $\call{G}=\set{Sig,Ver,Bl,Ubl, PK, SK}$ which satisfies the  following set of equations:

 $
 \call{H}=\left\lbrace\begin{array}{ll}
 \mbox{$Ver(Sig(x,SK(y)),PK(y))=x,$}\\
 \mbox{$Ubl(Bl(x,y),y)=x,$}\\
 \mbox{$Ubl(Sig(Bl(x,y),SK(z)),y)=Sig(x,SK(z)).$}
 \end{array}\right.
 $

Let $\call{R}$ be the set of rules obtained by orienting  equations of $\call{H}$ from left to  right,  $\call{R}$
is convergent and
it is obvious that any basic narrowing derivation \cite{Hullot80} issuing from any of the right hand side term of the rules of $\call{R}$ terminates.
This implies that any
narrowing derivation (and in particular basic narrowing derivation) issuing from any term terminates \cite{Hullot80} and thus 
$\call{H}$
has finite variant property \cite{Comon-LundhD05}.

The initial  deduction system  is given by the tuple $\call{I}_0=\intrus{\call{G}}{\call{L}_0}{\call{H}}$ and  we have:

 $
 \call{L}_0=\left\lbrace\begin{array}{ll}
 \mbox{$1:~ x,y\ded Sig(x,y),$}\\
 \mbox{$2:~ x,y\ded Ver(x,y),$}\\
 \mbox{$3:~ x,y\ded Bl(x,y),$}\\
\mbox{$4:~ x,y\ded Ubl(x,y).$}\\
 \end{array}\right.
 $

The first step of saturation outputs the following set of deduction rules:

$
\call{L}=\call{L}_0\cup
\left\lbrace\begin{array}{ll}
\mbox{$5:~ Sig(x,SK(y)),PK(y)\ded x,$}\\
\mbox{$6:~ Bl(x,y),y\ded x,$}\\
\mbox{$7:~ Sig(Bl(x,y),SK(z)),y\ded Sig(x,SK(z)).$}\\
\end{array}\right.
$

We define a new deduction system $\call{I}=\intrus{\call{G}}{\call{L}}{\emptyset}$
and by  lemma \ref{lemma:00:saturation1}, we have: $t\in\rhclos{E}{\call{I}_0}$ iff $t\in\rhclos{E}{\call{I}}$ for every set of ground terms $E$ (resp. a ground term $t$) in normal form.
From now we remark that the equational theory employed is the empty one.

Now, let us apply the second step of saturation. 
The closure applied on rules $1$ and $5$ outputs the rule $8:~ x,SK(y),PK(y)\ded x$, the closure applied on rules $3$ and  $6$ outputs the rule $9:~ x,y\ded x$ which will be deleted by the simplification step introduced above as consequence of lemma \ref{lem:elimination}.
The closure applied on rules $1$ and $7$ outputs the rule $10:~ y,Bl(x,y),SK(z)\ded Sig(x,SK(z))$.

We prove in the next lemma that the last rule is redundant when the  employed equational theory is the empty one.
\begin{lemma}\label{lemma:elimination-redundant-rules}
Let $\call{L}'_1=\call{L}\cup\set{x,SK(y),PK(y)\ded x}\cup\set{y,BL(x,y),SK(z)\ded Sig(x,SK(z))}$ and let $\call{L}'_2=\call{L}'_1\setminus\set{y,BL(x,y),SK(z)\ded Sig(x,SK(z))}$.
Suppose that the employed equational theory is the empty one. 
For any two sets of ground terms in normal form $E$ and $F$ we have:
$E\ded^*_{\call{L}'_2}F$ iff $E\ded^*_{\call{L}'_1}F$.  
\end{lemma}
\begin{proof}
Let $E$ and $F$ be two sets of normal ground terms.
The direct implication is obvious, let us prove the second one.
Suppose that $E\ded^*_{\call{L}'_1}F$ and let us prove that $E\ded^*_{\call{L}'_2}F$.
Suppose that in the $\call{L}'_1$-derivation $D$ starting from $E$ to $F$  there is some steps  where the applied rule is in $\call{L}'_1\setminus\call{L}'_2$ that is, by definition of $\call{L}'_1$ and $\call{L}'_2$, the applied rule is $y,Bl(x,y),Sk(z)\ded Sig(x,SK(z))$.

Let $i$ be the first step in the derivation where the applied rule is $y,Bl(x,y),Sk(z)\ded Sig(x,SK(z))$, we prove that this step can be replaced by other steps where the respectives applied rules are in $\call{L}'_2$.
$D:~ E=E_0\ded\ldots\ded E_i\ded _{y,Bl(x,y),Sk(z)\ded Sig(x,SK(z))}E_{i+1}\ldots \ded F$.
There is a ground substitution $\sigma$
in normal form such that $\set{y\sigma,Bl(x,y)\sigma,SK(z)\sigma}\subseteq E_i$ and $E_{i+1}=E_i\cup Sig(x\sigma,SK(z)\sigma)$.
Thus, the rule $Bl(x,y),y\ded x\in\call{L}'_2$ with the substitution $\sigma$ can be applied first on $E_i$ and outputs $E_{i1}=E_i\cup x\sigma$, then the rule $x,y\ded Sig(x,y)\in\call{L}'_2$ also with the substitution $\sigma$ can be applied on $E_{i1}$ and outputs $E_{i1}\cup{Sig(x\sigma,SK(z\sigma))}=E_{i+1}$.
We deduce that each application of the rule $y,Bl(x,y),Sk(z)\ded Sig(x,SK(z))$ in $D$  can be replaced by the application of two rules in $\call{L}'_2$.
We conclude that 
 $E\ded^*_{\call{L}'_1}F$ implies $E\ded^*_{\call{L}'_2}F$.
\end{proof}

\paragraph{Remarks.} 
\begin{description}

\item[Enforcing the termination of the Saturation.] 
The application of the \textit{Saturation} algorithm as is described in section \ref{sec:sat}
does not terminate. In fact, the rule $10$ is an increasing one and  \textit{closure} rule can be applied on rules $10$ and $7$. The application of the closure  outputs the rule $11:~ y,y',Bl(Bl(x,y),y'),SK(z)\ded Sig(x,SK(z))$
which is increasing.  We remark that \textit{closure} rule can be applied  on the rules $11$ and  $7$ and this application outputs a new  increasing rule. In addition,   \textit{closure} rule can be applied again on the new obtained rule and the rule $7$. We remark also that each such application of closure rule outputs a new increasing rule where the size of the terms in the left hand side is increased and closure rule can be applied again on this new obtained rule and the rule $7$.  
This implies that we have an infinite sequence of application of \textit{closure} rule.
We remark that this infinite sequence is due to the presence of the rule $10$.

As a consequence from the previous lemma (where we prove that the rule $y,Bl(x,y),SK(z)\ded Sig(x,SK(z))$ is redundant), we can delete this rule from the system immediately after its creation.
This deletion enforces the termination of the \textit{Saturation}.
\item[Saturated deduction system.]
Let $\II=\intrus{\call{G}}{\call{L}'}{\emptyset}$ be the saturated deduction system, we have:

$\call{L}'=\call{L}_0\cup\left\lbrace\begin{array}{ll}
\mbox{$Sig(x,SK(y)),PK(y)\ded x,$}\\
\mbox{$Bl(x,y),y\ded x,$}\\
\mbox{$Sig(Bl(x,y),SK(z)),y\ded Sig(x,SK(z)),$}\\
\mbox{$x,SK(y),PK(y)\ded x.$}\\
\end{array}\right.$

In $\call{L}'$, we note that only $\call{L}_0$-rules are increasing and the others are decreasing (by definition of \textit{increasing} and \textit{decreasing} rules).
\end{description} 

We recall that a derivation $D$ starting from $E$ of goal $t$ is
 \emph{well-formed} if for all rules $l\ded r$ applied with substitution
 $\sigma$, for all $u\in l\setminus\Variables$ we have either $u\sigma\in E$ or
 $u\sigma$ was constructed by a former decreasing rule.

In the next lemma, we prove that the  system $\call{L}'$ satisfies the following lemma.
\begin{lemma}
 Let $E$ (resp. $t$) be a set of terms (resp. a term) in normal form such that
  $t\in \rhclos{E}{\II}$. For all $\call{I}'$-derivations $D$ starting from $E$
  of goal $t$ we have either $D$ is well-formed or there is another
  $\call{I}'$-derivation $D'$ starting from $E$ of goal $t$ such that
  $\Cons{D}\subseteq\Cons{D'}$ and $D'$ is well-formed.
\end{lemma}
\begin{proof}
 We have $t\in\rhclos{E}{\II}$ implies that the set $\Omega(E,t)$ of
   \II-derivations starting from $E$ of goal $t$ is not empty.  Let
   $D\in\Omega(E,t)$, $D:E=E_0\ded E_1\ded \ldots \ded E_{n-1},t$, we denote
   $l_i\ded r_i$ the rule applied at step $i$ with the substitution $\sigma_i$: this rule is well-applied if for all $u\in l_i\setminus\Variables$, we have either $u\sigma\in E$ or $u\sigma$ was obtained by a former decreasing  rule, otherwise it  is bad-applied.

Suppose that $D$ is not well-formed then  there is at least one step in the derivation $D$ where the applied rule is bad-applied.   At 
each such step, one the following rule is applied: 

$
\left\lbrace\begin{array}{ll}
\mbox{$Sig(x,SK(y)),PK(y)\ded x,$}\\
\mbox{$Bl(x,y),y\ded x,$}\\
\mbox{$Sig(Bl(x,y),SK(z)),y\ded Sig(x,SK(z)),$}\\
\end{array}\right.
$ 

We note that the rule $x,SK(y),PK(y)\ded x$ can not be applied at such step because the rules $x\ded SK(x)$ and $x\ded PK(x)$ are not in $\call{L}'$.

Let us prove that each  application of the first  (resp. the second) rule  in $D$ such that there is a non variable term in left hand side of the rule where the instance is obtained by a former increasing rule can be deleted from $D$ without altering $\Cons{D}$.
Let $i$ be  the first step where the first  (resp. the second) rule  is bad applied, that is  there is a non variable term in left hand side where the instance is obtained by a former increasing rule. There is only one non variable term in the left hand side of the first (resp. the second) rule which can be obtained by a former increasing rule, this term is   $Sig(x,SK(y))$ (resp. $Bl(x,y)$).
Since the instance of this term, $Sig(x,SK(y))\sigma$ (resp. $Bl(x,y)\sigma$), is obtained by a former increasing rule  this last rule  will be $x,y\ded Sig(x,y)$ (resp. $x,y\ded Bl(x,y)$) and let $h$ ($h<i$) be the step where this rule is applied.
We deduce that $\set{x\sigma,SK(y\sigma)}  ~ (resp. ~ \set{x\sigma,y\sigma})\subseteq E_h$ and then   the rule applied at step  $i$ (which adds $x\sigma$) does not add a new term and the step $i$ can be deleted without modifying in $\Cons{D}$. Let $D'$ be the obtained derivation, we have $\Cons{D}'=\Cons{D}$.
We deduce that every step in $D'$ where the rule $Sig(x,SK(y)),PK(y)\ded x$ (resp. the rule $Bl(x,y),y\ded x$) is bad applied   can be deleted without altering in the  trace of $D'$ and let  $d$ be the obtained derivation.
We note that  every application of the rule $Sig(x,SK(y)),PK(y)\ded x$ (resp. the rule $Bl(x,y),y\ded x$) in $d$ is a well-application.

Suppose that $d$ is not well-formed then there is at least one step where the  rule applied  is bad-applied. Let $i$ be the first such step then the rule applied is  $Sig(Bl(x,y),SK(z)),y\ded Sig(x,SK(z))$ and $Sig(Bl(x,y),SK(z))\sigma$ is obtained by a former increasing rule, $x,y\ded Sig(x,y)$. Let $h$, ($h<i$), be the step where this increasing rule  is applied.
We deduce that $\set{Bl(x,y)\sigma,SK(z)\sigma}\subseteq E_h$. If $x\sigma\notin E_i$ then the rule applied at step $i$ in $d$ can be replaced  first by  the application of $Bl(x,y),y\ded x$ then the application of $x,y\ded Sig(x,y)$.
Let $d'$ be the obtained derivation,  $d':~ E\ded \ldots \ded E_i\ded_{Bl(x,y),y\ded x} E_i,x\sigma\ded _{x,y\ded Sig(x,y)} E_i,x\sigma, Sig(x,SK(z))\sigma\ded \ldots \ded E_{n-1},t$.
By above and since  $x\sigma\notin E_i$ we have either $Bl(x,y)\sigma\in E$ or $Bl(x,y)\sigma$ is obtained by a former decreasing rule.

If $x\sigma\in E_i$  then the rule applied at step $i$ in $d$ can be replaced by  the application of $x,y\ded Sig(x,y)$.
Let $d''$ be the obtained derivation,  $d'':~ E\ded \ldots \ded E_i\ded _{x,y\ded Sig(x,y)} E_i,Sig(x,SK(z))\sigma\ded \ldots\ded E_{n-1},t$.

This implies that each bad application of the rule $Sig(Bl(x,y),SK(z)),y\ded Sig(x,SK(z))$ can be replaced by one (or  two) well-applied rules. We deduce that if the derivation $D$ is not well-formed there is another well-formed derivation $D''$ starting from $E$ of goal $t$ such that $\Cons{D}\subseteq\Cons{D''}$. 
\end{proof}

We remark that the above lemma is similar to the lemma \ref{lemma:00:saturation2}.

In order to solve $\call{I}_0$-ground reachability problems (definition \ref{def:groud.constraint.system}), we apply the algorithm defined in section \ref{sec:decision}.
Since the saturation applied on $\call{L}_0$ terminates, by lemmas (\ref{lemma:00:lem2}, \ref{lemma:00:lem3}, \ref{lemma:00:completeness}, \ref{lemma:00:correcteness} and \ref{lem:terminaison-ID}) we deduce the following corollary:   
\begin{corollary}
The $\call{I}_0$-ground reachability problem is decidable.
\end{corollary}

\section{Decidability of reachability problems for subterm convergent theories}\label{sec:subterm}

In this section, we give a decidability result for the reachability problems for a class of subterm convergent equational theories. We recall that subterm convergent equational theories have finite variant property \cite{Comon-LundhD05}.
The result of this section is entailed by a more general result by Baudet \cite{Baudet05}, but the proof here in this specific case is much shorter.

We recall that $\call{G}$ is a set of functions symbols and we denote by
$\call{H}$ a subterm convergent equational theory and by
$\call{I}_0=\intrus{\call{G}}{\call{L}_0}{\call{H}}$ the initial deduction  system
such that $\call{L}_0$ is the union of functions $x_1,\ldots,x_n\ded
  f(x_1,\ldots,x_n)$ for some function symbols $f\in\call{G}$.

\begin{definition}(Subterm convergent theories.)
An equational theory \call{H} is subterm convergent if it is generated by a convergent rewriting system \call{R} and for each rule $l\to r\in\call{R}$, $r$ is a strict subterm of $l$.
\end{definition}

In the rest of this section, we give an algorithm to decide the following reachability problem:
\begin{decisionproblem}{$\call{I}_0$-Reachability Problem}
  \entreeu{ An $\call{I}_0$-constraint system \call{C}.}
  \sortie{\textsc{Sat} iff there exists a substitution $\sigma$ such
    that $\sigma\models_{\call{I}_0}\call{C}.$}
\end{decisionproblem}

We let $\II=\intrus{\call{G}}{\call{L}'}{\emptyset}$ to be the saturated deduction system. 
We suppose that $r\notin l$ for all  rules $l\ded r\in\call{L}'$ that is rules not satisfying this property will be deleted.

In the following lemma we prove that, in the case of subterm convergent equational theories and under our assumption on the form of initial deduction rules $\call{L}_0$,  \textit{Saturation} terminates and the obtained  new rules are  decreasing.

\begin{lemma}{\label{lem:subterm}}
  The saturation of $\call{L}_0$ terminates and for every rule $l\ded r\in\call{L}'\setminus\call{L}_0$ there exists a term $s\in l$ such that
  $r$ is a strict subterm of $s$.
\end{lemma}
\begin{proof}
Let $l\ded r\in\call{L}'\setminus\call{L}_0$
and let us prove that this rule satisfies the following property: there is a term $s\in l$ such that $r\in\SSub{s}$.  By induction on the number of saturations needed to obtain a rule $l\ded r$.

  Let us first prove this property  is true for rules obtained by the step  1 of the saturation.  By definition of $\call{H}$, by the fact that variants of term are in normal form and  given the assumption that
  all original rules are $x_1,\ldots,x_n\ded f(x_1,\ldots,x_n)$, this implies:
  $$
  \begin{array}{c@{\hspace*{3em}}rcl}
    &   \norm{f(x_1,\ldots,x_n)\theta}\in\SSub{f(x_1,\ldots,x_n)\theta}\\
    \text{Thus, there exists }i\text{ such that:}&
    \norm{f(x_1,\ldots,x_n)\theta}\in\Sub{x_i\theta}\\
  \end{array}
  $$
   If there is equality, the rule is removed (since  $r\notin l$ for all rules $l\ded r$). This implies that all rules obtained from step 1 of saturation satisfies the property. 
Since $\call{L}_0$ is finite and since  subterm convergent equational theories have finite variant property \cite{Comon-LundhD05},  first step of saturation terminates. Since $u\in\SSub{v}$ implies $u\prec v$, rules obtained by step 1 are decreasing.
Let $\call{L}$ be the set of rules obtained by step 1 and let us prove that rules obtained by closure satisfy the property.
Let us prove it for the first  rule obtained by closure.
By definition of closure rule and since rules in $\call{L}\setminus\call{L}_0$ are decreasing, the first closure will be applied on rules $x_1\ldots,x_n\ded f(x_1,\ldots,x_n)\in\call{L}_0$ and $f(s_1,\ldots,s_n),l\ded r\in\call{L}\setminus\call{L}_0$. Again by definition of closure, the obtained rule is   $s_1,\ldots,s_n,l\ded r$. By definition of decreasing rule, 
there is a term $u\in\set{f(s_1,\ldots,s_n),l}$ such that $r\in\SSub{u}$, if $u=l$ then the new rule satisfies  the property  and if  $u=f(s_1,\ldots,s_n)$ then there is an integer $i$ such that $r\in\Sub{s_i}$. If $r\in\SSub{s_i}$ the obtained rule satisfies the property  else the rule can not be in $\call{L}'$ (since rules $l\ded r$ with $r\in l$ are deleted).
We conclude that the first rule obtained by closure is decreasing and if we apply again closure, it will be applied on a rule in $\call{L}_0$ and a rule not in $\call{L}_0$. 
 We conclude that rules obtained by step 2 satisfy the property and are decreasing . We conclude also that step 2 terminates. 
\end{proof}

We recall that increasing rules are of form $x_1,\ldots,x_n\ded f(x_1,\ldots,x_n)$ for a function symbol $f\in\call{G}$ (Lemma \ref{lem:subterm}).

\subsection{Decidability result}
We recall that our goal is to solve $\call{I}_0$-reachability problem.
\paragraph{Algorithm.}
Let $\call{C}^0=((E^0_i\rhd{}v^0_i)_{i\in\set{1,\ldots,n}},\call{S}^0)$.
\begin{quotation}
  \emph{Step~1.} Guess a finite
variant substitution $\theta$ for all terms of $\call{C}^0$, apply $\theta$ on these terms and normalise them then solve the obtained unification system. Finally, apply the obtained solution $\alpha$ on the constraints. Let $\call{C}=((E_i\rhd{}t_i)_{i\in\set{1,\ldots,n}})$ be the obtained constraint system.
\end{quotation}

We remark that  this step terminates 
and it is also  correct (Lemma \ref{lemma:00:lem3}) and  complete (Lemma \ref{lemma:00:lem2}). Unless otherwise specified, $\II$ is the deduction system implicit in all notations in the rest of this section.

 We now introduce the notation $\rhd_{inc}$ to denote a deduction
 constraint that has to be solved using only increasing  rules. We say
 a constraint $E\rhd_{inc} t$ is in solved form if $t$ is a variable. The
 constraint system is in solved form if all the deduction constraints
 are in solved form. The application of a decreasing  rule
 $l\ded r$ on a constraint $E\rhd t$ is defined as follows, and in
 accordance with Lemma~\ref{lemma:00:saturation2}:
 \begin{itemize}
 \item let $\sigma$ be the mgu of the terms in
   $l\setminus\Variables$ with a subset $F$ of $E\setminus\Variables$
 \item if $\set{x_1,\ldots,x_k} = l\cap\Variables$, replace
   $\call{C}_\alpha,E\rhd t,\call{C}_\beta$ with:
 $$
 (\call{C}_\alpha, E\rhd_{inc} x_1,\ldots,E\rhd_{inc} x_k,
 E\cup\set{r} \rhd t ,\call{C}'_\beta)\sigma
 $$
 Where $\call{C}'_\beta$ is constructed from $\call{C}_\beta$ by adding
 $r$ to each left-hand side. This last construction aims
 at preserving the inclusion of knowledge sets.
 \end{itemize}

\begin{quotation}
  \emph{Step 2.} Iterate until the constraint system is in solved form
  or unsolvable:
  \begin{enumerate}
  \item Put all tagged deduction constraints $E\rhd_{inc} t$ in solved
    form;
  \item If all constraints preceding an untagged $E\rhd t$ are in
    solved form, Apply non-deterministically
    $|\Sub{E}\setminus\Var{E}| $ decreasing  rules on $E$. Replace
    $E\rhd t$ by the obtained deduction constraints, all tagged with
    $inc$.
  \end{enumerate}
\end{quotation}

Let us prove the completeness and termination of Step~2.

\paragraph{Completeness.}  The proof of the following lemma is trivial
by the form  of increasing rules.

\begin{lemma}{\label{lem:subterm:composition}}
  If $\sigma\models E\rhd_{inc} f(t_1,\ldots,t_n)$ then either
  $f(t_1,\ldots,t_n)\sigma \in E\sigma$ or
  $x_1,\ldots,x_n\ded f(x_1,\ldots,x_n)$ will be in $\call{L}_0$ and  for each
  $i\in\set{1,\ldots,n}$ we have $\sigma\models E\rhd_{inc} t_i$.
\end{lemma}

The first part of the iteration consists either in transforming a
deduction constraint $E\rhd_{inc} f(t_1,\ldots,t_n)$ into $E\rhd_{inc}
t_1,\ldots, E\rhd_{inc} t_n$, or in unifying $f(t_1,\ldots,t_n)$ with
$e\in E$. By Lemma~\ref{lem:subterm:composition}, given a ground
substitution $\sigma$ such that $\sigma\models E\rhd_{inc}
f(t_1,\ldots,t_n)$ there exists a sequence of choices reducing
$E\rhd_{inc} f(t_1,\ldots,t_n)$ to a (possibly empty) set of deduction constraints
$E\tau\rhd_{inc} u_1,\ldots E\tau \rhd_{inc} u_k$ where the $u_1,\ldots,u_k$
are variables or constants. If there is a constant which is not in $E\tau$
the constraint is not satisfiable (by definition of increasing
rules), and the sequence of choices fails.

Let us now consider the second part of the iteration.

\begin{lemma}{\label{lem:variable:subterm}}
  Assume $\sigma\models E\rhd_{inc} x$ with $x$ the first variable in the
  sequence of deduction constraints such that $t\in\Sub{x\sigma}$ for
  some ground term $t$.  Then either there exists $u\in\Sub{E}$ such
  that $u\sigma=t$ or $t\in\rhclos{E\sigma}{\call{L}'_{inc}}$.
\end{lemma}

\begin{proof}
  Let us assume there does not exist $u\in\Sub{E}$ such that
  $u\sigma=t$. By minimality of $x$ and the determinacy of constraint
  systems we have $t\notin\Sub{\Var{E}\sigma}$. Since $\Sub{E\sigma} =
  \Sub{E}\sigma\cup\Sub{\Var{E}\sigma}$ we have $t\notin\Sub{E\sigma}$
  and, by hypothesis on $x$ and $t$, $t\in\Sub{x\sigma}$. Since
  $\sigma\models E\rhd_{inc} x$ consider a derivation
  $E_1=E\sigma\ded\ldots\ded E_{n-1}\cup x\sigma$, and let $i$ be minimal
  such that $t\in\Sub{E_i}$. The index $i$ exists since
  $t\in\Sub{x\sigma}$, and is different from $1$
  since $t\notin\Sub{E\sigma}$. By definition of the increasing rules
  we then must have $E_i=E_{i-1},t$.
\end{proof}

Consider a $\II$-constraint system $\call{C}=(\call{C}_\alpha,E\rhd t,\call{C}_\beta)$ satisfied by a substitution $\sigma$ and 
 all deduction constraints in $\call{C}_\alpha$ are in solved
form. By Lemmas \ref{lemma:00:saturation2} and ~\ref{lem:subterm} and by the fact that $r\notin l$ for all rules $l\ded r \in\call{L}'$, all  decreasing rules applied
on $E\sigma$ yield a term in \Sub{E\sigma}. Thus there are at most
$|\Sub{E}\setminus\Var{E}|$ different terms that can be obtained by
decreasing rule starting from $E\sigma$ and which are not in
$\Sub{\Var{E}\sigma}$.  Assume a term $t$ is in
$\Sub{\Var{E}\sigma}\setminus\Sub{E}\sigma$, and let $x$ be the first
variable (in the ordering of deduction constraints) such that
$t\in\Sub{x\sigma}$.  By definition of constraint systems there exists
a deduction constraint $E_x\rhd_{inc} x$ in $\call{C}_\alpha$.  Since
$E_x\subseteq E$, by Lemma~\ref{lem:variable:subterm}, we have
$t\in\rhclos{E_x\sigma}{\call{L}'_{inc}}$. Again, since $E_x\sigma \subseteq
E\sigma$, this implies $t\in\rhclos{E\sigma}{\call{L}'_{inc}}$: the decreasing rule 
was not useful, and can be replaced by a sequence of increasing. Thus
 in \clos{E\sigma} at most  $|\Sub{E}\setminus \Var{E}|$ terms 
are deducible using decreasing   rules.  Thus, after a right
choice of at most $|\Sub{E}\setminus \Var{E}|$ decreasing rules,
all terms deducible from the obtained knowledge set can be deduced
using only increasing  rules, hence the tagging with $inc$ of the final
deduction constraint $E\cup\set{r_1,\ldots,r_k}\rhd_{inc} t$, $k=|\Sub{E}\setminus \Var{E}|$.

\paragraph{Termination of Step~2.}  First let us notice that if a
unification is chosen, it unifies two subterms of the constraint
system in the empty theory, and thus either the two terms were already
equal or it reduces strictly the number of variables in the constraint
system. Thus the number of unification choices is bounded by the
number of variables in the constraint system. Once all unification
have been performed, the termination of the first part of the
iteration can easily be proved by considering the multiset of the
right-hand side of the deduction constraints, ordered by the extension
to multisets of the (well-founded) subterm ordering. The second part
of the iteration obviously terminates. Thus each iteration terminates.
Since each iteration decreases strictly the number of non-labelled
deduction constraints, Step~2.  terminates.

\section{Conclusion}

In~\cite{fossacs04}, H.~Comon-Lundh proposes a two-steps strategy for solving
general reachability problems: first, decide ground reachability problems and,
second, reduce general reachability problems to ground reachability ones,
\textit{e.g.} by providing a bound on the size of a minimal solution of a
problem. Our results are in this line: for \emph{contracting} deduction systems,
general reachability can be reduced to ground reachability. We strongly
conjecture that it permits one to provide a bound on the size of minimal
solutions. Also, we leave to the reader the proof of the fact that if saturation
terminates, the deduction system is \emph{local} in the sense defined
in~\cite{BernatC06}. Thus, this paper adds a new criterion to the one already known
for deciding reachability problems.

In future works, we will investigate how the construction presented here can be
extended to equational theories having the finite variant property w.r.t. a
non-empty equational theory. We will also try to weaken the definition of $\mu(T)$ for a set of terms $T$.

\bibliographystyle{plain}
\bibliography{biblio}

\begin{thebibliography}{10}

\bibitem{spore}
{S}ecurity {P}rotocols {O}pen {R}epository.
\newblock http://www.lsv.ens-cachan.fr/spore/.

\bibitem{AbadiG99}
M.~Aadi and A.D. Gordon.
\newblock A calculus for cryptographic protocols: The spi calculus.
\newblock {\em Information and Computation}, pages 148(1):1--70, Jan. 1999.

\bibitem{AbadiR00}
M.~Aadi and P.~Rogaway.
\newblock Reconciling two views of cryptography (the computational soundness of
  formal encryption).
\newblock {\em In Proc. 1st IFIP International conference on Theorectical
  Computer Science (IFIP-TCS), LNCS}, 1872:3--22, Springer--Verlag, 2000.

\bibitem{AbadiC05}
Mart\'{\i}n Abadi and V{\'e}ronique Cortier.
\newblock Deciding knowledge in security protocols under (many more) equational
  theories.
\newblock In {\em CSFW}, pages 62--76. IEEE Computer Society, 2005.

\bibitem{AbadiC06}
Mart\'{\i}n Abadi and V{\'e}ronique Cortier.
\newblock Deciding knowledge in security protocols under equational theories.
\newblock {\em Theor. Comput. Sci.}, 367(1-2):2--32, 2006.

\bibitem{AmadioL00}
Roberto~M. Amadio and Denis Lugiez.
\newblock On the reachability problem in cryptographic protocols.
\newblock In Catuscia Palamidessi, editor, {\em CONCUR}, volume 1877 of {\em
  Lecture Notes in Computer Science}, pages 380--394. Springer, 2000.

\bibitem{AmadioLV03}
Roberto~M. Amadio, Denis Lugiez, and Vincent Vanack{\`e}re.
\newblock On the symbolic reduction of processes with cryptographic functions.
\newblock {\em Theor. Comput. Sci.}, 290(1):695--740, 2003.

\bibitem{Baudet05}
Mathieu Baudet.
\newblock Deciding security of protocols against off-line guessing attacks.
\newblock In Vijay Atluri, Catherine Meadows, and Ari Juels, editors, {\em ACM
  Conference on Computer and Communications Security}, pages 16--25. ACM, 2005.

\bibitem{BernatC06}
Vincent Bernat and Hubert Comon-Lundh.
\newblock Normal proofs in intruder theories.
\newblock In Mitsu Okada and Ichiro Satoh, editors, {\em ASIAN}, volume 4435 of
  {\em Lecture Notes in Computer Science}, pages 151--166. Springer, 2006.

\bibitem{BurrowsAN90}
Michael Burrows, Mart\'{\i}n Abadi, and Roger~M. Needham.
\newblock A logic of authentication.
\newblock {\em ACM Trans. Comput. Syst.}, 8(1):18--36, 1990.

\bibitem{ChevalierK07}
Yannick Chevalier and Mounira Kourjieh.
\newblock Key substitution in the symbolic analysis of cryptographic protocols.
\newblock In Vikraman Arvind and Sanjiva Prasad, editors, {\em FSTTCS}, volume
  4855 of {\em Lecture Notes in Computer Science}, pages 121--132. Springer,
  2007.

\bibitem{ChevalierKRT-FSTTCS03}
Yannick Chevalier, Ralf K{\"u}sters, Micha{\"e}l Rusinowitch, and Mathieu
  Turuani.
\newblock Deciding the security of protocols with diffie-hellman exponentiation
  and products in exponents.
\newblock In Paritosh~K. Pandya and Jaikumar Radhakrishnan, editors, {\em
  FSTTCS}, volume 2914 of {\em Lecture Notes in Computer Science}, pages
  124--135. Springer, 2003.

\bibitem{ChevalierKRT05}
Yannick Chevalier, Ralf K{\"u}sters, Micha{\"e}l Rusinowitch, and Mathieu
  Turuani.
\newblock An np decision procedure for protocol insecurity with xor.
\newblock {\em Theor. Comput. Sci.}, 338(1-3):247--274, 2005.

\bibitem{ChevalierLR07}
Yannick Chevalier, Denis Lugiez, and Micha{\"e}l Rusinowitch.
\newblock Towards an automatic analysis of web service security.
\newblock In Boris Konev and Frank Wolter, editors, {\em FroCos}, volume 4720
  of {\em Lecture Notes in Computer Science}, pages 133--147. Springer, 2007.

\bibitem{ChevalierR05}
Yannick Chevalier and Micha{\"e}l Rusinowitch.
\newblock Combining intruder theories.
\newblock In Lu\'{\i}s Caires, Giuseppe~F. Italiano, Lu\'{\i}s Monteiro,
  Catuscia Palamidessi, and Moti Yung, editors, {\em ICALP}, volume 3580 of
  {\em Lecture Notes in Computer Science}, pages 639--651. Springer, 2005.

\bibitem{CJ97}
J.~Clark and J.~Jacob.
\newblock A survey of authentication protocol literature.

\bibitem{fossacs04}
Hubert Comon-Lundh.
\newblock Intruder theories (ongoing work).
\newblock In Igor Walukiewicz, editor, {\em FoSSaCS}, volume 2987 of {\em
  Lecture Notes in Computer Science}, pages 1--4. Springer, 2004.

\bibitem{Comon-LundhD05}
Hubert Comon-Lundh and St{\'e}phanie Delaune.
\newblock The finite variant property: How to get rid of some algebraic
  properties.
\newblock In J{\"u}rgen Giesl, editor, {\em RTA}, volume 3467 of {\em Lecture
  Notes in Computer Science}, pages 294--307. Springer, 2005.

\bibitem{CortierDL}
V.~Cortier, S.~Delaune, and P.~Lafourcade.
\newblock A survey of algebraic properties used in cryptographic protocols.
\newblock {\em Research Report LSV-04-15, LSV, ENS de Cachan}, Sept. 2004.

\bibitem{DershowitzJ90}
Nachum Dershowitz and Jean-Pierre Jouannaud.
\newblock Rewrite systems.
\newblock In {\em Handbook of Theoretical Computer Science, Volume B: Formal
  Models and Sematics (B)}, pages 243--320. 1990.

\bibitem{dolev83ieee}
Danny Dolev and Andrew Chi-Chih Yao.
\newblock On the security of public key protocols.
\newblock {\em IEEE Transactions on Information Theory}, 29(2):198--207, 1983.

\bibitem{Hullot80}
Jean-Marie Hullot.
\newblock Canonical forms and unification.
\newblock In Wolfgang Bibel and Robert~A. Kowalski, editors, {\em CADE},
  volume~87 of {\em Lecture Notes in Computer Science}, pages 318--334.
  Springer, 1980.

\bibitem{KremerR05}
Steve Kremer and Mark Ryan.
\newblock Analysis of an electronic voting protocol in the applied pi calculus.
\newblock In Shmuel Sagiv, editor, {\em ESOP}, volume 3444 of {\em LNCS}, pages
  186--200. Springer, 2005.

\bibitem{Lowe95}
G.~Lowe.
\newblock An attack on the needham-schroeder public key authentication
  protocol.
\newblock {\em Information processing letters}, 1995.

\bibitem{Lowe97}
Gavin Lowe.
\newblock Casper: A compiler for the analysis of security protocols.
\newblock In {\em CSFW}, pages 18--30. IEEE Computer Society, 1997.

\bibitem{Lowe99a}
Gavin Lowe.
\newblock Towards a completeness result for model checking of security
  protocols.
\newblock {\em Journal of Computer Security}, 7(1), 1999.

\bibitem{MartelliM82}
Alberto Martelli and Ugo Montanari.
\newblock An efficient unification algorithm.
\newblock {\em ACM Trans. Program. Lang. Syst.}, 4(2):258--282, 1982.

\bibitem{meadows96}
Catherine Meadows.
\newblock The nrl protocol analyzer: An overview.
\newblock {\em J. Log. Program.}, 26(2):113--131, 1996.

\bibitem{MillenS01}
Jonathan~K. Millen and Vitaly Shmatikov.
\newblock Constraint solving for bounded-process cryptographic protocol
  analysis.
\newblock In {\em ACM Conference on Computer and Communications Security},
  pages 166--175, 2001.

\bibitem{MMS97}
J.C. Mitchell, M.~Mitchell, and U.~Stern.
\newblock Automated analysis of cryptographic protocols using mur$\phi$.
\newblock In {\em Proc. IEEE Symposium on Research in Security and Privacy},
  pages 141--153. IEEE Computer Society Press, 1997.

\bibitem{NarendranPS97}
Paliath Narendran, Frank Pfenning, and Richard Statman.
\newblock On the unification problem for cartesian closed categories.
\newblock {\em J. Symb. Log.}, 62(2):636--647, 1997.

\bibitem{RusinowitchT-TCS03}
Micha{\"e}l Rusinowitch and Mathieu Turuani.
\newblock Protocol insecurity with a finite number of sessions, composed keys
  is np-complete.
\newblock {\em Theor. Comput. Sci.}, 1-3(299):451--475, 2003.

\end{thebibliography}

\end{document}